\documentclass[draftcls, 12pt, onecolumn, romanappendices]{IEEEtran} 

\IEEEoverridecommandlockouts                            
\usepackage{amsmath} %
\usepackage{amssymb}  %

\title{\LARGE \bf Regulation, Volatility and Efficiency in Continuous-Time Markets}

\usepackage[noadjust]{cite}
\usepackage{enumerate}
\usepackage{graphicx}

\newtheorem{thm}{Theorem}[section] 
\newtheorem{lem}[thm]{Lemma}
\newtheorem{defn}{Definition}[section] 
\newtheorem{prop}[thm]{Proposition} 

\newtheorem{cor}{Corollary}

\newcounter{taskcounter}
\newenvironment{hypot}{ \refstepcounter{taskcounter}\textbf{A\arabic{taskcounter}}:}
{}

\newcommand {\ass}[1]{\textbf{A\ref{#1}}}

\def \t{^\top}
\def \b{\mathbf}
\def \tr{\mathrm{Tr}}
\def \teq{\triangleq}
\def \th{\mathbf{\theta}}
\def \Th{\mathbf{\Theta}}

\def \be{\begin{equation}}
\def \ee{\end{equation}}
\def \bml{\begin{multline}}
\def \eml{\end{multline}}
\def \ba{\begin{align}}
\def \ea{\end{align}}

\def \qquadtwo{\qquad\qquad}
\def \qquadthree{\qquad\qquad\qquad}

\def \qquadeight{\qquad\qquad\qquad\qquad\qquad\qquad\qquad\qquad}

\def \R{\mathbb{R}}		%
\def \F{\mathcal{F}}     	%
\def \w{\omega}                 %
\def \conv{\rightarrow}		%
\def \e{\epsilon}             	%
\def \p{\partial}		%
\def \E{\mathbb{E}}		%
\def \l{\lambda}		%
\def \f{\infty}			%

\author{Arman~C. Kizilkale and Shie Mannor%

\thanks{Arman~C. Kizilkale is with the Department of Electrical and Computer Engineering at McGill University, Montreal, Canada. {\tt\small arman@cim.mcgill.ca}}%
\thanks{Shie Mannor is with the Department of Electrical Engineering at the Technion, Haifa, Israel.  {\tt\small shie@ee.technion.ac.il}}%
}

\begin{document}

\maketitle
\thispagestyle{empty}

\begin{abstract}
We analyze the efficiency of markets with friction, particularly power markets. We model the market as a dynamic system with $(d_t;\,t\geq 0)$ the demand process and $(s_t;\,t\geq 0)$ the supply process. Using stochastic differential equations to model the dynamics with friction, we investigate the efficiency of the market under an integrated expected undiscounted cost function solving the optimal control problem. Then, we extend the setup to a game theoretic model where multiple suppliers and consumers interact continuously by setting prices in a dynamic market with friction. We investigate the equilibrium, and analyze the efficiency of the market under an integrated expected social cost function. We provide an intriguing \emph{efficiency--volatility} no-free-lunch trade-off theorem.
\end{abstract}

\section{INTRODUCTION}\label{sec:intro}

The first attempts of privatization and deregulation of power industry took place in the 1980s starting in Chile and the UK \cite{2006Jo}. After the restructuring of power markets in California in the late 1990s, price fluctuations have resulted in an estimate of \$45 billion in higher electricity costs, lost businesses due to long blackouts, and a weakening economic growth according to the Public Policy Institution of California \cite{2005Ro}. Even though such events have been mostly considered as market failures \cite{1998WNS,2003NS}, it was shown in \cite{2010CM} that the occurrence of choke-up prices (the maximum price a consumer is willing to pay) is intrinsic to markets with \emph{friction}, and the market mechanism is efficient in a stylized model. Choke-up prices are observed in current market mechanisms regardless of being efficient, intrinsic, or market failure; this is undesirable and costly.

Dynamic pricing in electricity markets have interesting characteristics. Locational Marginal Pricing (LMP) schemes may determine very high prices for a region while a neighbour region might be assigned a low or even a negative price for the same amount of energy, where the supplier is actually willing to pay to consumers for the power they use. The constraints due to transmission congestion, voltage and thermal constraints, Kirchoff's Laws and start-up and shut-down costs are the main reasons behind excess and lack of supply which cause volatile prices \cite{1992Ho}.

As shown in \cite{2010CM}, the current deregulated market mechanism is efficient with respect to the infinite horizon social cost. However, the definition of \emph{efficiency} depends on the social cost function defined. Many models do not penalize \emph{volatility} in their cost functions. One can define \emph{volatility} as rapid or unexpected changes in the price process. Several models of deterministic and stochastic volatility have been studied in the economics literature including the famous deterministic Black-Scholes formula \cite{1973BS}, and the stochastic Heston's extension \cite{1993He}, SABR \cite{2002HKLW} and GARCH \cite{1986Bo} models. We are going to adopt a much simpler definition of volatility since our goal is to give a specialized analysis of volatility in power markets.

Several authors studied efficiency in power markets. Even though most studies are based on static frameworks, it was shown in \cite{2008Ma} that under ramping constraints, markets might face prices not necessarily equal to the marginal cost price. A dynamic game model based on duopoly markets is analyzed in \cite{2009GS_ET}, and a dynamic competitive equilibrium for a stochastic market model is formulated and the role of volatility for the value of wind generation is presented in \cite{2010MPWKS_CDC}.

We model the power market through continuous dynamics and an integrated undiscounted cost function. The problem is presented as an optimal control problem, and the control action is defined as an increment process applied by the regulator. The HJB equation is solved and the resulting optimal control is presented. As a special case, in the class of linear quadratic cost functions, we analytically show that there is a trade-off between efficiency and non-volatility. In the second part of the paper we take a decentralized approach and define the market as a dynamic linear quadratic game among individual decision maker supplier and consumer agents. The agents are coupled through the price process. We show that this price process can be estimated, and the agents can calculate their best response actions based on this estimation. We show that these best response actions constitute an equilibrium and the trade-off theorem between efficiency and non-volatility is shown to hold in this dynamic game model as well.

In the first part of the paper we suggest a dynamic optimization framework for power markets: in Sec.\ \ref{sec:model}, we introduce the model we are going to use for the centralized control model. Demand $( d_t;\,t\geq 0)$, supply $( s_t;\,t\geq 0)$ and price $( p_t;\,t\geq 0)$ processes are defined for the social cost optimizer regulator agent $R$ with the corresponding cost function. In Sec.\ \ref{sec:CCF}, we present the optimal control that leads to a volatile price process. In Sec.\ \ref{sec:NFL}, we define volatility and modify the social cost function to account for it. We solve the dynamic stochastic optimization problem for linear dynamics and a quadratic cost function and present the closed form solution. We show that there is a trade-off that can be quantified between efficiency and non-volatility, and present supporting simulations. In the second part of the paper we suggest a dynamic game-theoretic optimization framework: in Sec.\ \ref{sec:DCF}, the consumer agents $D_i,\, 1\leq i \leq N^d$, with their dynamics $( d_t^i;\, t\geq 0)$, the supplier agents $S_i,\, 1\leq i \leq N^s$ with their dynamics $( s_t^i;\, t\geq 0)$, and the price process $( p_t;\,t\geq 0)$ are defined with the corresponding cost functions for the consumers and suppliers. In this framework there is no regulator agent: the price process is solely determined in the market mechanism through the actions of the consumers and suppliers \cite{2010KM_All}. In Sec.\ \ref{sec:equ_ana}, we first show the existence of best response actions for the game model, we present the closed form solutions, and finally we analyze the equilibrium properties of the system. In Sec.\ \ref{sec:eff_vol}, we define volatility for this model, and show that the trade-off theorem can be extended to the multi-player game setup. We present supporting simulations in Sec.\ \ref{sec:sim} and conclude in Sec.\ \ref{sec:Conc}.

\section{MODEL}\label{sec:model}
In this section we define the optimization problem for the social cost optimizer in power markets. Here we call the optimizer the ``regulator'' (agent $R$). We define the three dimensional state process $(x_t: x_t = (d_t,s_t,p_t)\t;\,t\geq 0)$. We have $(d_t;\,t\geq 0)$, the demand process, $(s_t;\,t\geq 0)$, the supply process, and $(p_t;\,t\geq 0)$, the price process. 
Demand and supply dynamics are defined as
\begin{equation}\label{eqn:dynamics_cont} 
\begin{aligned}
dd_t =& f^d(d_t,p_t)dt + \sigma_d dw_t^d, \quad t\geq 0,\\
ds_t =& f^s(s_t,p_t)dt + \sigma_s dw_t^s, \quad t\geq 0,
\end{aligned}
\end{equation} 
using deterministic continuous functions $f^d$ and $f^s$ with $(w_t^d;\,t\geq 0)$ and $(w_t^s;\,t\geq 0)$, standard Wiener processes. The function $f^d$ is allowed to be a function of $d$ and $p$, values of demand and price, and $f^s$ is allowed to be a function of $s$ and $p$, values of supply and price processes.

We employ the following assumptions on the functions $f^d$ and $f^s$ in \eqref{eqn:dynamics_cont}. The first assumption \ass{ass:fric} reflects \emph{friction} for power markets. This assumption ensures that the instantaneous change in demand and supply processes with respect to a price change is constrained. This is one of the key properties of power dynamics: the suppliers and consumers are unable to respond to abrupt changes in the system instantly. The reason for the supplier's sluggishness is slow ramp up in power production, whereas for the consumers it is usually not handy or very complicated and costly to startup and shutdown a running machine or a household. The second assumption, \ass{ass:incr_dyna_func}, reflects natural characteristics of demand and supply dynamics: demand is a decreasing function of the price, whereas supply is an increasing function of the price.

\begin{hypot}\label{ass:fric}
For constant $C_1>0$, $f^d(0,0)\leq C_1,\; f^s(0,0) \leq C_1$ and
\[ \left \lvert \frac{\p f^d}{\p p} \right \rvert + \left \lvert \frac{\p f^s}{\p p}\right \rvert  \leq C_1  .\] 
\end{hypot}
An immediate example is a linear function of the form $f(x,p) = A(t)x + B(t)p$ with $A,\, B$ of class $\b{C}^1([0,T])$. 

\begin{hypot}\label{ass:incr_dyna_func}
$f^d$ is a strictly decreasing function of $p$, whereas $f^s$ is strictly increasing. 
\end{hypot}

This assumption ensures that an increase in price is reflected on the deterministic portion of decreasing demand dynamics and increasing supply dynamics.

We also adopt the assumption below for initial values of the processes and the disturbance process:

\begin{hypot}\label{ass:init}
$\{d_0,s_0,p_0\in  \R\}$ are mutually independently distributed bounded initial conditions, and $\{w^d,w^s\}$ are mutually independent and independent of the initial conditions. Instantaneous variances of the disturbance processes, $\sigma_d^2,\sigma_s^2$, are bounded. 
\end{hypot}

We adopt the stepwise price adjustment model \cite{2004HCM_TAC} for the optimizer (so called regulator agent $R$), where the bounded input control process $(u_t;\, t\geq 0 )$ controls the amount of the increment. The price process controlled by agent $R$'s input is defined as
\begin{equation}\label{eqn:control_cont}
dp_t = u_t dt, \quad |u_t|\leq u_{max}, \quad t\geq0.
\end{equation}

The actions of $R$ is the set $\{ u:\lvert u \rvert < u_{max}, u\in \R, u_{max}>0\}$ which is simply the constrained price adjustment. $R$ observes the demand and supply processes and taking into consideration their dynamics, cost function and the constraint on price increment, takes an action in terms of increasing or decreasing the power price. This action is intended to control market dynamics by only applying increments on the price process. 

Following \cite{2010CM}, the individual loss functions of the consumer and supplier are defined respectively:
\begin{align*}
g^d(d,s,p) &= p \cdot s - v\cdot \min(d,s) + c_{bo}(r), \\
g^s(d,s,p) &= c(s)-p \cdot s.
\end{align*}
Here, $c(s)\in \b{C}_b^2:\R \rightarrow \R_+$, with polynomial growth with power $k_1$ or less, i.e., $\lvert c(s) \rvert \leq \lvert c \rvert (1+s^{k_1})$, where $\b{C}_b^2$ denotes the family of all bounded functions which are twice differentiable. The function $c(s)$ is the production cost, and is strictly convex and strictly increasing with respect to $s$. One needs to work on a realistic production cost function in order to have a reasonable power market model. We note that in real power markets, production cost is not a convex function. The startup and shutdown costs, transmission line constraints, weather fluctuations all affect the production cost function. However, if one neglects the startup and shutdown costs, the cost function resembles a convex function \cite[see Figure 1]{2008Ma}. For our model we will assume a continuous convex cost. The constant $v\in(0,\infty)$ is the value the consumer obtains for a unit of power. The blackout is denoted by $c_{bo}(r)\in \mathbf{C}_b^2:\mathbb{R}\rightarrow \mathbb{R}_+$, with polynomial growth with power $k_2$ or less, i.e., $\lvert c_{bo}(r)\rvert \leq \lvert c_{bo} \rvert (1+r^{k_2})$, is convex, zero on $[0,\infty)$ and strictly decreasing on $(-\infty,0)$, where $r$ denotes the reserve, $r := s - d$. In other words, if the total consumption in the system can not be met, blackout cost is paid.
In the spot market, the consumer, $D$, pays $p \cdot s$, the price of all the supply bought, to the supplier, $S$. Note that $v$ is multiplied by the supplied portion of his demand. Blackout cost $c_{bo}(\cdot)$ is a function of the unmet demand. Further note that the supplier $S$ pays for all the cost of production, and gains unit price multiplied with all the units of supply bought by the consumer agent $D$. Finally, we employ the following integrated expected social cost function that is simply the sum of the consumer $D$ and the supplier $S$ loss functions integrated in time:
\begin{equation}\label{eqn:cost_int_cont} 
J(x,u) = \mathbb{E}\int_0^T \left[ -v\cdot \min(d_t,s_t)+c(s_t)+c_{bo}(r_t)\right] dt.
\end{equation}
In the section that follows, we consider the optimality of the cost function presented above with the dynamics \eqref{eqn:dynamics_cont}, the control \eqref{eqn:control_cont} and the cost function \eqref{eqn:cost_int_cont} under \ass{ass:fric}, \ass{ass:incr_dyna_func} and \ass{ass:init}.

\section{CENTRALIZED CONTROL FORMULATION}\label{sec:CCF}
In this section we analyze the optimal control problem in terms of the state vector $x:=(d,s,p)\t$. As stated before, this is a centralized control problem for the regulator agent $R$. In principle $R$'s objective is to regulate demand and supply processes using the price increments as the control tool, so that the best social outcome is achieved. In this section we show that the optimal control of the regulator is a ``bang-bang" control, which leads to volatile prices. We write (\ref{eqn:dynamics_cont}) and (\ref{eqn:control_cont}) in vector form with stochastic dynamics as
\begin{equation}\label{eqn:system_CCF}
dx = \psi dt + Gd\omega,\quad t\geq 0,
\end{equation}
where $\w$ is a $3\times 1$ standard Wiener process. We set $x := (d,s,p)\t,\psi := (f\t(\cdot),u)\t$ and,
\begin{equation*}
f(d,s,p) = \left( \begin{array}{c}
f^{d}(d,p)\\
f^{s}(s,p) \end{array} \right), \enspace G = \left( \begin{array}{ccc}
\sigma_d & 0 & 0\\
0 & \sigma_s & 0\\
0 & 0 & 0 \end{array} \right).
\end{equation*}

The loss function of \eqref{eqn:cost_int_cont} is rewritten here as $ g(x) = g(d,s,p) = -v\cdot \min(d,s)+c(s)+c_{bo}(s-d)$. The admissible control for the regulator is specified as $\mathcal{U}=\{u(\cdot):u$ adapted to $\sigma(x_s,s \leq t)$ and $u(t)\in U = [-u_{max},u_{max}], t\geq 0\}$. Therefore, the regulator can at most increase or decrease the price with unit $u_{max}$ and $-u_{max}$ at each iteration. Finally, the cost associated with \eqref{eqn:system_CCF} and a control $u$ is specified to be $J(x_0,u)=\mathbb{E}[\int_0^T g(d_t,s_t,p_t)dt]$. Further, we set the value function
\begin{equation}\label{eqn:V_function}
V(0,x_0) \enspace \triangleq \enspace \inf_{u\in\mathcal{U}}J(x_0,u).
\end{equation}
The theorem that follows claims the existence and uniqueness of the optimal control to the problem \eqref{eqn:V_function}.

\begin{thm}\label{thm:CCF_exis}
There exists a unique $\hat{u}\in\mathcal{U}$ such that $J(x_0, \hat{u}) = \inf_{u\in\mathcal{U}}J(x_0,u)$, where $x_0 = (d_0,s_0,p_0)\t$ is the initial state at time $t_0=0$, and if $\tilde{u}\in\mathcal{U}$ is another control such that $J(x_0,\tilde{u})=J(x_0,\hat{u})$, then $\mathbb{P}_{\Omega}(\tilde{u}_s\neq \hat{u}_s)>0$ only on a set of times $s\in[0,T]$ of Lebesgue measure zero. 
\end{thm}
\textit{Proof:} The proof is given in Appendix \ref{thm:CCF_exis_pro}.

Now that we have shown the existence and uniqueness of a control, we check for approaches to compute the optimal solution. For a function class $\mathcal{G}$: (i) $V \in \b{C}([0,T]\times \R^3)$, (ii) $ \lvert V \rvert \leq C_v(1+d^{k_1}+s^{k_2})$ where $C_v,k_1,k_2$ depend on $V$, (iii) $V(T,x)=0$, we write the HJB Equation
\begin{equation}\label{eqn:HJB_CCF} 
-\frac{\p V}{\p t} + \sup_{u \in \mathcal{U}} \left \{-\frac{\p\t V}{\p x}\psi \right \} - \frac{1}{2}\tr \left (\frac{\p^2 V}{\p x^2}GG\t \right ) - g(\cdot) = 0.
\end{equation} 
A classical solution to the HJB Equation \eqref{eqn:HJB_CCF} does not exist as $GG\t$ is not of full rank in \eqref{eqn:system_CCF} \cite{1999BO}. Therefore, viscosity solutions are adopted.

\begin{defn}\textit{Viscosity solution:} \cite[Sec.\ 4, Def.\ 5.1]{1999YZ}

A function $\underline{v}(t,x) \in \b{C}([0,T]\times \R^3)$ is a viscosity subsolution to the HJB equation \eqref{eqn:HJB_CCF} if $\underline{v}|_{t=T}\leq 0$, and for any $\phi(t,z)\in \b{C}^{1,2}([0,T]\times \R^3)$, whenever $\underline{v}-\phi$ obtains a local maximum at $(t,x)\in [0,T)\times \R^3$, we have
\begin{equation}\label{eqn:vis}
-\frac{\p \phi}{\p t} + \sup_{u \in \mathcal{U}} \left \{-\frac{\p\t \phi}{\p x}\psi \right \} - \frac{1}{2}\tr \left (\frac{\p^2 \phi}{\p x^2}GG\t \right ) - g(\cdot) \leq 0.
\end{equation} 
A function $\overline{v}(t,x)\in \b{C}([0,T]\times \R^3)$ is called a viscosity supersolution to \eqref{eqn:HJB_CCF} if $\overline{v}|_{t=T}\geq 0$, and whenever $\overline{v}-\phi$ takes a local minimum at $(t,x)\in[0,T)\times \R^3$, in \eqref{eqn:vis} the inequality is changed to $``\geq"$. A value function $v(t,x)$ is a viscosity solution if it is a viscosity subsolution and a viscosity supersolution. 
\end{defn}

\begin{thm}\label{thm:CCF_HJBexis}
The value function defined in \eqref{eqn:V_function} is the unique viscosity solution to the HJB equation \eqref{eqn:HJB_CCF} in the class $\mathcal{G}$.
\end{thm}
\begin{proof} H4.1 and H4.2 of \cite{2005HCM_SIAM} are satisfied. Theorem 4.1 of \cite{2005HCM_SIAM} proves that  $V$ defined in \eqref{eqn:V_function} is a viscosity solution to the HJB equation \eqref{eqn:HJB_CCF}, and Theorem 4.3 of \cite{2005HCM_SIAM} proves that the solution is a unique solution to \eqref{eqn:HJB_CCF} in the class $\mathcal{G}$.
\end{proof}

\subsection{Perturbation Method}\label{sec:CCF_PM}

In order to make the $GG\t$ matrix full rank, we add $(1/2)\epsilon^2(\p^2V/\p p^2)$ to \eqref{eqn:HJB_CCF} \cite{1975FR}. For a function class $\mathcal{G}'$: (i) $V \in \b{C}^{1,2}([0,T]\times \R^3)$, (ii) $ \lvert V \rvert \leq C_v(1+d^{k_1}+s^{k_2})$ where $C_v,k_1,k_2$ depend on $V$, (iii) $V(T,x)=0$, we write the HJB Equation
\begin{multline}\label{eqn:HJB_exp}
-\frac{\p V^p}{\p t} - \frac{\p V^p}{\p d}f^d(d,p) -   \frac{\p V^p}{\p s}f^s(s,p) + \sup_{u \in \mathcal{U}} \left \{-\frac{\p V^p}{\p p}u \right \} \\
-\frac{1}{2}\sigma_d^2\frac{\p^2 V^p}{\p d^2} - \frac{1}{2} \sigma_s^2 \frac{\p^2 V^p}{\p s^2} - \frac{1}{2}\epsilon^2\frac{\p^2 V^p}{\p p^2} - g(d,s,p) = 0,
\end{multline}
where $V^p(T,x)=0$.

\begin{lem}\label{lem:GS} \cite[Sec.\ 6, Theorem 4]{1972GS}
For each $k=1,2,...$
\[ \E |x(t)|^k \leq C_k (1+\E |x(s)|^k),\quad s\leq t \leq T, \]
where the constant $C_k$ depends on $k,\, T-s$, and $\psi$.  
\end{lem}

\begin{lem}\label{lem:CCF_exisparder} \cite[Lemma 6.2]{1975FR}
Let $B\subset \R^3$ be bounded, $V^p$ a solution of \eqref{eqn:HJB_exp} in $\b{C}^{1,2}((0,T)\times \R^3)$ with $V^p$ continuous in $\b{C}^{1,2}([0,T]\times \R^3)$ and $V^p(T,x)=0$. Then there exists a constant $M_B$ such that 
\[ |V^p(t,x)| \leq M_B,\quad |V_t^p|\leq M_B \quad \text{for all } x\in B,\, 0\leq t < T,\]
where the constant $M_B$ depends only on $B,\, T$, the constant $C_1$ in \ass{ass:fric} and $c$, $c_{bo}$ defined for $c(s)$ and $c_{bo}(r)$.
\end{lem}

\begin{thm}\label{thm:CCF_HJBperexis}
The perturbed HJB equation \eqref{eqn:HJB_exp} has a unique classical solution in the class $\mathcal{G}'$ for all $\e>0$.
\end{thm}
\begin{proof}
We employ an approximation approach. Let us first take $0 < \e < 1$. For integer $d\geq 1$, let $h^d(x)$ be such that $h^d(x)=1$ for $\lvert x \rvert \leq d$, $h^d(x)=0$ for $\lvert x \rvert \geq d+1$, and $\lvert h^d_x \rvert \leq 2$. Let $V^d$ be the solution to
\begin{multline}\label{eqn:HJB_expaux}
-\frac{\p V^d}{\p t} - \frac{\p V^d}{\p d}f^d(d,p)h^d(x) -   \frac{\p V^d}{\p s}f^s(s,p)h^d(x) + \sup_{u \in \mathcal{U}} \left \{-\frac{\p V^d}{\p p}u \right \} h^d(x)\\
-\frac{1}{2}\sigma_d^2\frac{\p^2 V^d}{\p d^2} - \frac{1}{2} \sigma_s^2 \frac{\p^2 V^d}{\p s^2} - \frac{1}{2}\epsilon^2\frac{\p^2 V^d}{\p p^2} - g(d,s,p) h^d(x)= 0,
\end{multline}
where $V^d(T,x)=0$.

For fixed $d_0>1$ and $D=(0,T)\times (\lvert x \rvert < d_0)$, for any $d\geq d_0$, $V^d(t,x)$ satisfies \eqref{eqn:HJB_exp} for $|x|<d_0$. Lemma \ref{lem:CCF_exisparder} ensures that $V^d,\, V^d_d,\, V^d_s, \, V^d_p$ are uniformly bounded on $D$. For any $D'=(0,T)\times (|x|<d')$, $0<d'<d_0$, by local estimates
\[
\lVert V^d \rVert_{\l,D}^{(2)} \teq \lVert V^d \rVert_{\l,D} + \lVert V_t^d \rVert_{\l,D}+\sum_{i=1}^3 \lVert V_{x_i}^d\rVert_{\l,D} + \sum_{i,j=1}^3 \lVert V_{x_i,x_j}^d \rVert_{\l,D} \]
is uniformly bounded, where $\lVert \cdot \rVert_{\l,D}$ denotes a Sobolev type $L^\l(D)$ norm, where $L^\l(K)$ denotes the space of $\l$-th power integrable functions on $K\subset Q$. Take $\l>3$, and by the H\"older estimates, $V_{x_i}^d$ satisfies a uniform H\"older condition on any compact subset of $D'$. Moreover, $V_t^d,V_{x_i,x_j}^d,d=d_0+1,d_0+2,...,$ satisfy a uniform H\"older condition on such a $D'$. At this point we employ Arzela-Ascoli theorem and take a subsequence $\{d_{k_q};q\geq 1\}$ such that $V^{d_{k_q}},V_t^{d_{k_q}},V_{x_i}^{d_{k_q}},V_{x_i,x_j}^{d_{k_q}}$ converge uniformly to $V^p,V_t^p,V_{x_i}^p,V_{x_i,x_j}^p$ on $D'$, respectively, as $q\conv\f$, where $V^p$ satisfies \eqref{eqn:HJB_exp} and is in the class $\mathcal{G}'$ due to the growth condition on $g$ and the compactness of $U$. In the next theorem, we use the It\=o's formula to show that $V^p$ is the value function to a related stochastic control system, and thus it is a unique solution to \eqref{eqn:HJB_exp} in the class $\mathcal{G}'$. 

\end{proof} 
\begin{thm}\label{thm:CCF_HJBperconv}
Let $x\in \R^3$ and $\epsilon>0$. Define $V^p$ as solution to \eqref{eqn:HJB_exp} and $V$ as solution to \eqref{eqn:HJB_CCF} for the admissible control set $\mathcal{U}$. Then $V^p\rightarrow V$ uniformly on $[0,T]$.
\end{thm}
\begin{IEEEproof}
For $(v_t;\,t\geq 0)$ a standard Wiener process, we can define an alternative control action in the form of a stochastic differential equation $dp_t^p = u_t dt + \epsilon dv_t$. The resulting value function can be shown to be a viscosity solution to \eqref{eqn:HJB_exp}, and this solution is unique (see Chapter 4, \cite{1999YZ}). For a fixed $u$, we have $P\{\lim_{\epsilon\rightarrow 0}\sup_{0\leq t\leq T} \lvert p^p - p\rvert =0 \} = 1$. We recall from \eqref{eqn:V_function} that $V(t,x)$ is the infimum of $J(x,u)$ among non-anticipative controls in $\mathcal{U}$. Let $k_1,k_2$ be as in the polynomial growth conditions for $c(\cdot)$ and $c_{bo}(\cdot)$, Since $\mathcal{U}$ is compact, \ass{ass:fric} together with Lemma \ref{lem:GS} imply that $\E|x(t)|^{k_1}$ and $\E|x(t)|^{k_2}$ are bounded uniformly with respect to $t\in[0,T)$ and $u\in\mathcal{U}$. It follows that $J^p(x,u)$ is uniformly bounded. One can use Lebesgue's dominated convergence theorem to obtain $\lvert J^p(x,u) - J(x,u) \rvert \rightarrow 0, \, \text{as } \epsilon\rightarrow 0$, and $V^p\rightarrow V$ as $\epsilon\rightarrow 0$ follows. By adopting Arzela-Ascoli Theorem similar to the methodology that was employed in the proof of Theorem \ref{thm:CCF_HJBperexis}, one can obtain $V^p\rightarrow V$ uniformly on $[0,T]$, as $\epsilon\rightarrow 0$.
\end{IEEEproof}
This gives us the following result:
\begin{cor}
For the function class $\mathcal{G}'$ the solution $u^* \in \mathcal{U}$ to the perturbed HJB Equation \eqref{eqn:HJB_exp} is found as:
\begin{equation}\label{eqn:vopt_sol}
u^* = \arg\min_{u\in\mathcal{U}}\frac{\p\t V^p}{\p x}\psi=-\mathrm{sgn}\left(\frac{\p V^p}{\p p}\right) u_{max}, 
\end{equation}
where $u$ was previously defined as $dp_t = u_t dt,\, t\geq0,\, |u_t|\leq u_{max}$.
\end{cor}

When we look at the the perturbed HJB Equation \eqref{eqn:HJB_exp}, the bound $ \lvert V \rvert \leq C(1+d^{k_1}+s^{k_2})$ is a direct estimate, the value function is differentiable everywhere in the function class $\mathcal{G}'$, and due to the constraint defined on the control action, the optimal control is represented as a bang-bang control. Hence, the optimal control is found as a single switch.  At the boundary we have $V(T,x)=0$. Therefore, one can numerically solve \eqref{eqn:HJB_exp}. 

In Theorem \ref{thm:CCF_exis} we showed the existence of an optimal control to the problem \eqref{eqn:V_function}. Due to the problematic nature of stochastic differential equations, we have seen that the solution of an optimal control in ``classical sense" may not exist. This leads us to formulate a suboptimal approach. The convergence of the suboptimal solution to the optimal solution was shown. 

The control is shown to be a simple single switch. This has significant consequences, i.e., we proved that the regulator needs to increase the price increment to the possible maximum or decrease it to the possible minimum depending on the value obtained from \eqref{eqn:vopt_sol}. Due to \ass{ass:fric}, the effect of price on demand and supply is constrained. Therefore, a certain amount of time is needed in order to adjust the levels of demand and supply in the system. For cases where demand is much bigger than supply or supply is much bigger than demand, the maximal increment has to be applied for a long period of time. Hence, volatile prices are the optimal outcome of the market with respect to the cost function \eqref{eqn:cost_int_cont}.

Note that \ass{ass:fric} is important both for technical reasons and for modeling reasons. In addition to the fact that \ass{ass:fric} models friction, if \ass{ass:fric} is removed, the polynomial growth of the value function also may not be satisfied. Moreover, for a hypothetical frictionless market, a single increment on price would adjust demand and supply levels to the desired levels instantly; thus, less volatility would be expected. Indeed, for a completely deterministic frictionless system, volatility would be zero.

\section{EFFICIENCY--VOLATILITY TRADE-OFF}\label{sec:NFL}
Non-volatility and efficiency are two desirable properties of power markets. In this section we show that these two notions contradict each other in a market model with friction. Therefore, one has to trade-off non-volatility and efficiency in designing the market mechanism.

The optimal control policy for the system \eqref{eqn:dynamics_cont} and the price process due to the nature of the optimal control \eqref{eqn:vopt_sol} were discussed in the previous section. Since the demand and supply processes are defined by stochastic differential equations, they fluctuate on their trajectories and the regulator modifies the price process for the optimal outcome. The highest cost is paid when the difference between demand and supply is the highest.

In this section we prove that no efficient regulation strategy can exist that maintains a smooth price process when supply and demand are defined by mean-reverting stochastic differential equations.

We form a function that penalizes the control action $u$. Recall the loss function defined in \eqref{eqn:cost_int_cont}. We adopt the stepwise price adjustment model defined in \eqref{eqn:control_cont}, where the input control process $(u_t;\,t \geq 0)$ controls the amount of the increment. The cost associated with the system is defined as 
\begin{equation}\label{eqn:NFL_int_cost}
J(x_0,u) = \mathbb{E}\int_0^T  [ g(d_t,s_t,p_t) + r  u_t^2 ] dt,
\end{equation}
where we add $r u^2$ to the term \eqref{eqn:cost_int_cont} and $r>0$ is the \emph{volatility coefficient}. We will prove that if the \textit{volatility coefficient} decreases, the expected cost decreases. In other words, if high volatility is not allowed, the social cost defined in \eqref{eqn:cost_int_cont} increases. 

We define \emph{efficiency} as the quantity obtained when the expected cost is multiplied by -1 taken out the control action penalizing part: $-\mathbb{E}\int_0^T  [ g(d_t,s_t,p_t)  ] dt$. \emph{Volatility} on the other hand is defined by the price fluctuation measured by $\mathbb{E}\int_0^T  u_t^2  dt$.

We require one more assumption here:

\begin{hypot}\label{ass:OU}
The supply process $(s_t;\,t\geq 0)$ and the demand process $(d_t;\,t\geq 0)$ are linear mean-reverting processes that have bounded variances and admit stationary probability distributions in case of time invariant means.
\end{hypot}
As a special case, we study a linear quadratic cost function of the form
\begin{equation}\label{eqn:NFL_int_app_cost}
J(x_0,u) = \mathbb{E}\int_0^T (x_t\t Qx_t + 2x_t\t D + r u_t^2) dt,
\end{equation}
where $x := (d,s,p)\t$, and $Q\geq 0$, $r>0$ and $D$ are constant values. Employing \ass{ass:OU}, we have the dynamics
\begin{align}
\nonumber & dx_t=\psi(x_t,u_t) dt+G dw_t,\quad t\geq 0,\\
\label{eqn:lin_dynamics} &dx_t=\left( Ax_t+Bu_t+h \right) dt+G dw_t,\quad t\geq 0,
\end{align}
where $\omega$ is a $3\times 1$ standard Wiener process, 
$x(0)=x_0$, and $A,B,G$ are in the form of
\be\label{eqn:ABG}
A =  \left( \begin{array}{ccc}
\ast & 0 & \ast\\
0 & \ast & \ast\\
0 & 0 & 0 \end{array} \right),\quad B=\left( \begin{array}{c}
0\\0\\1\end{array} \right),\quad G = \left( \begin{array}{ccc}
\sigma_d & 0 & 0\\
0 & \sigma_s & 0\\
0 & 0 & 0 \end{array} \right),
\ee
where `$\ast$' denotes a bounded constant. 

\subsection{Existence and Uniqueness of the Optimal Control}

From now on, we will work on \eqref{eqn:NFL_int_app_cost} and \eqref{eqn:lin_dynamics}. We take the admissible control set $\mathcal{U}_2=\{u:u \text{ adapted to }\sigma(x_s,s\leq t) \text{ and }  \int_0^T u_t^2 dt < \infty\}$. The minimum cost-to-go from any initial state $(x)$ and any initial time $(t)$ is described by the \emph{value function} which is defined by $V(t,x) = \inf_{u\in \mathcal{U}_2}J(x,u)$. The optimal control problem is well defined with the Hamilton-Jacobi-Bellman (HJB) Equation
\begin{equation}\label{eqn:NFL_HJB}
-\frac{\p V}{\p t}  + \sup_{u \in \mathcal{U}_2} \left \{ -\frac{\p V}{\p x}\t \psi - r u^2 \right \} - \frac{1}{2}\tr \left ( \frac{\p^2 V}{\p x^2}GG\t \right ) - x\t Q x - 2x\t D = 0,
\end{equation}
where $V(T,x) = 0$.

As discussed earlier in Sec.\ \ref{sec:CCF}, due to the lack of uniform parabolicity, standard solutions may be hard to obtain. Viscosity solutions are adopted in these circumstances. Therefore we add the term $(1/2)\epsilon^2(\p^2V/\p p^2)$ to \eqref{eqn:NFL_HJB} and obtain uniform parabolicity. Equation \eqref{eqn:NFL_HJB} then becomes
\begin{equation}\label{eqn:NFL_per_HJB}
-\frac{\p V^p}{\p t}  + \sup_{u \in \mathcal{U}_2} \left \{ -\frac{\p V^p}{\p x}\t \psi - r u^2 \right \} - \frac{1}{2}\sigma_d^2\frac{\p^2 V^p}{\p d^2 } - \frac{1}{2}\sigma_s^2\frac{\p^2 V^p}{\p s^2 } - \frac{1}{2}\epsilon^2 \frac{\p^2 V^p}{\p p^2 } - x\t Qx - 2x\t D = 0,
\end{equation}
where $V^p(T,x)=0$.

Equation \eqref{eqn:NFL_per_HJB} has a unique solution as stated in the following theorem.
\begin{thm}\label{thm:CCF_LQRHJBperexis}
Equation \eqref{eqn:NFL_per_HJB} has a unique classical solution for the admissible control set $\mathcal{U}_2$ for all $\e>0$.
\end{thm}
\textit{Proof:} The proof is very similar to the Proof of Theorem \ref{thm:CCF_HJBperexis}, therefore omitted.

In the theorem below, we prove that the solution to the perturbed value function \eqref{eqn:NFL_per_HJB} converges uniformly to the value function obtained from the HJB Equation \eqref{eqn:NFL_HJB}.
\begin{thm}\label{thm:CCF_LQRHJBperconv}
Let $x\in \R^3$ and $\epsilon>0$. Define $V^p$ as solution to \eqref{eqn:NFL_per_HJB} and $V$ as solution to \eqref{eqn:NFL_HJB} for the admissible control set $\mathcal{U}_2$. Then $V^p\rightarrow V$ uniformly on $[0,T]$.
\end{thm}
The proof is very similar to the proof of Theorem \ref{thm:CCF_HJBperconv}, therefore omitted.

\subsection{Closed Form Solution}\label{sec:CCF_cl}

Standard arguments \cite[Section 2.3]{1989AM} show that $J(x,u)$ is quadratic in $x$. Furthermore, at any point $x\in \mathbb{R}^3$ and $t\in [0,T]$ the minimum cost-to-go is quadratic in $x$. Consequently, one can model $V$ of the form $V(t,x) = x\t K(t) x + 2 x\t S(t) + q(t)$ that satisfies the boundary condition $V(T,x) = 0,\enspace \forall x\in \mathbb{R}^3$. Substituting $V$ in \eqref{eqn:NFL_HJB} and applying first order optimization gives
\begin{equation}\label{eqn:NFL_opt_control}
u^*(t) = -r^{-1}B\t [K(t)x(t)+S(t)].
\end{equation}

Solving the closed loop expression we get the ODEs:
\begin{align}
& \label{eqn:K}\dot{K} + KA + A\t K - KBr^{-1}B\t K + Q = 0,\\
& \label{eqn:S}\dot S + (A - B r^{-1} B\t K )\t S + Kh  + D = 0,\\
& \label{eqn:q}\dot{q} + 2S\t h - S\t Br^{-1}B\t S + \tr(KGG\t) = 0,
\end{align}
with boundary conditions $K(T)=0$, $S(T)=0$ and $q(T)=0$.
The linear quadratic optimal control problem admits a unique optimum feedback controller given by \eqref{eqn:NFL_opt_control} which obtains the minimum value of the cost function $J(x_0,u^*) = x_0\t K(0) x_0 + 2x_0\t S(0) + q(0)$.

\subsection{Efficiency--Volatility Trade-off}

We would like to look at the relation between $r$, the \emph{volatility coefficient}, and the state penalizing part of the cost function obtained when the volatility term is removed from the cost function. We define the \emph{state penalizing cost} as
\begin{equation}\label{eqn:spcost}
J_{sp}^*(x_0,u^*) \enspace \triangleq \enspace \mathbb{E} \int_0^T \left [ x_t\t Qx_t + 2 x_t\t  D \right ]  dt,
\end{equation}
which is denoted as \emph{efficiency} when multiplied by $-1$.

\begin{thm}\label{thm:spc}
Suppose \ass{ass:fric}-\ass{ass:OU} hold. For all $x\in\R^3$, the state penalizing cost portion \eqref{eqn:spcost} of the cost function \eqref{eqn:NFL_int_app_cost} using optimal control $u^*$ is an increasing function of $r$.
\end{thm}

\textit{Proof:} The proof is presented in Appendix \ref{thm:spc_pro}.

Increasing the \emph{volatility coefficient} increases social cost, therefore decreases \emph{efficiency}, while decreasing the coefficient decreases the cost, hence increases \emph{efficiency}. On the other hand increasing the \emph{volatility coefficient} decreases volatility, whereas decreasing \emph{volatility coefficient} increases volatility. Therefore, there is a trade-off between social \emph{efficiency} and \emph{non-volatility}.

\subsection{Simulations}

\subsubsection{Analytical Supportive Simulation}

Here we simulate a power market. We use Euler-Maruyama Method \cite{1974Ar} for discretization of the stochastic differential equations. The dynamics equations are $d_{k+1} = d_k - \rho\left ( d_k-(\beta-p_k) \right ) \Delta t + \sigma w_k^d\sqrt{\Delta t},\enspace s_{k+1} = s_k - \rho\left ( s_k-(p_k-\gamma) \right ) \Delta t + \sigma w_k^s\sqrt{\Delta t},\enspace p_{k+1} = p_k + u_k \Delta t$, where $\rho=0.05, \Delta t = 0.05, \beta = 75, \gamma = 25,\sigma=2,t_{final}=100$, with the initial conditions $x_0 = (d_0,s_0,p_0)\t=(25,25,50)\t$. We use mean-reverting processes with time varying means. The power market we simulate consists of a demand process with mean $(75-p)$ MWh, and a supply process with mean $(p-25)$ MWh. Therefore for a price of \$50 per MWh, the supplier is expected to produce $25$ MW of power, whereas the demand in the system is also expected to be 25 MW. In accordance with \ass{ass:incr_dyna_func}, the demand is an decreasing function of price, whereas supply is increasing. We calculate ${dJ_{sp}}/{dr}$ using Theorem \ref{thm:spc} using a range of values of $r$ and present the result  in Fig.\ \ref{fig:dJdr_r}, and as expected, it is always positive. Also, as expected it is a convex function; the value is very high for small values of $r$ and converges to 0 as $r$ increases. Increasing $r$, the \emph{volatility coefficient}, corresponds to decreasing volatility which ends up with a cost increase as ${dJ_{sp}}/{dr}>0$ for all $r>0$. In Fig. \ref{fig:tradeoff} we present the trade-off between the efficiency and the non-volatility. The numbers are normalized, and one can see that in a market with higher volatility the efficiency is higher. Here, on the $x$ axis 0 corresponds to the situation where $r$ is very large and 1 corresponds to the situation where $r=0$. On the $y$ axis, the corresponding values are normalized, so that 0 is the lowest and 1 is the highest efficiency that can be obtained.

\begin{figure}[ht]
\begin{minipage}[b]{0.5\linewidth}
\centering
\includegraphics[width=9.25cm]{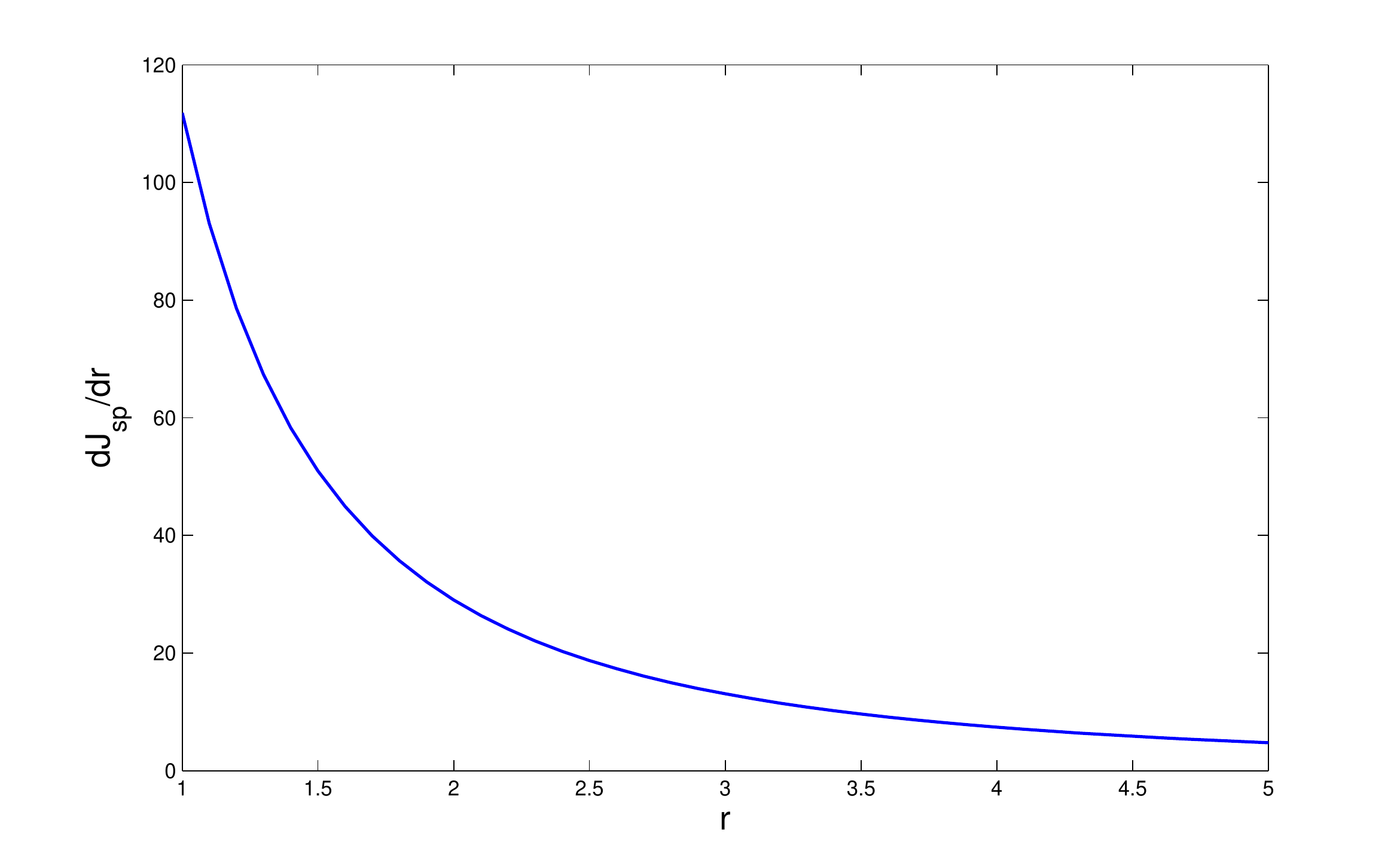}
\caption{$dJ_{sp}/dr$}
\label{fig:dJdr_r}
\end{minipage}
\begin{minipage}[b]{0.5\linewidth}
\centering
\includegraphics[width=9.25cm]{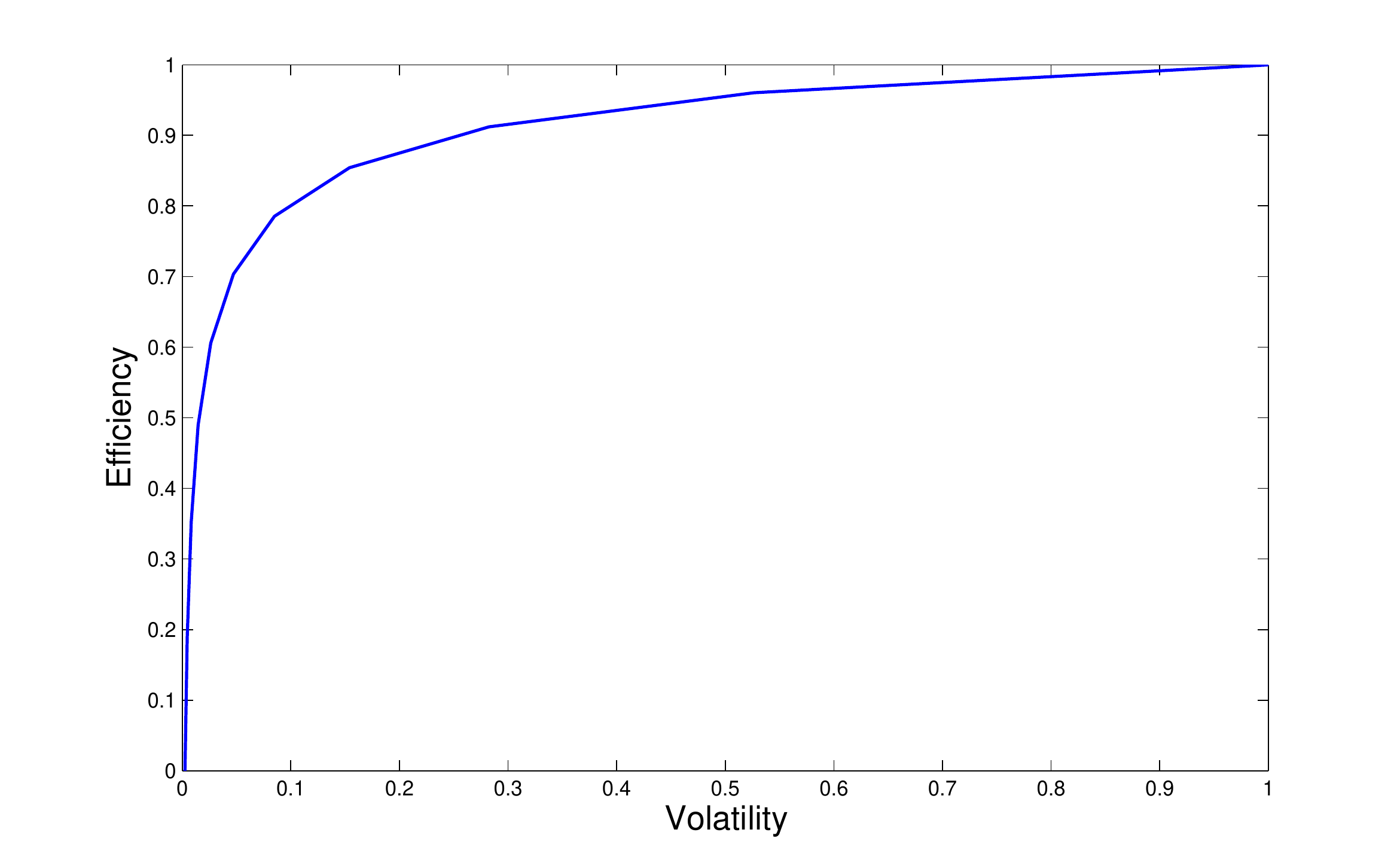}
\caption{Trade-off} 
\label{fig:tradeoff}
\end{minipage} 
\end{figure}

\subsubsection{Numerical Simulation}

Here we present a couple of simulations showing the dynamics when $r=0.01$ and $r=1000$. The high volatility in Fig.\ \ref{fig:dynamics_R01_C} compared to the low volatility in Fig.\ \ref{fig:dynamics_R1000_C} can be observed. One can also notice the effect of volatility on stability.

Also in Fig.\ \ref{fig:dynamics_R1000_C} the optimal actions of the regulator agent can be observed at 4 points on the trajectory. At P1, the demand goes up due to stochasticity and the regulator acts with full force to increase the price, so that stability can be obtained. At P2, price gets high, and the demand is taken under control; gradually the regulator decreases the price. Between 60 seconds and 80 seconds, we see that supply follows a higher level than the demand. The regulator acts to take the price down to a local minimum at P3. Then, until P4 the regulator gradually increases the price until it comes to a local maximum at P4.

\begin{figure}[ht]
\begin{minipage}[b]{0.5\linewidth}
\centering
\includegraphics[width=9.25cm]{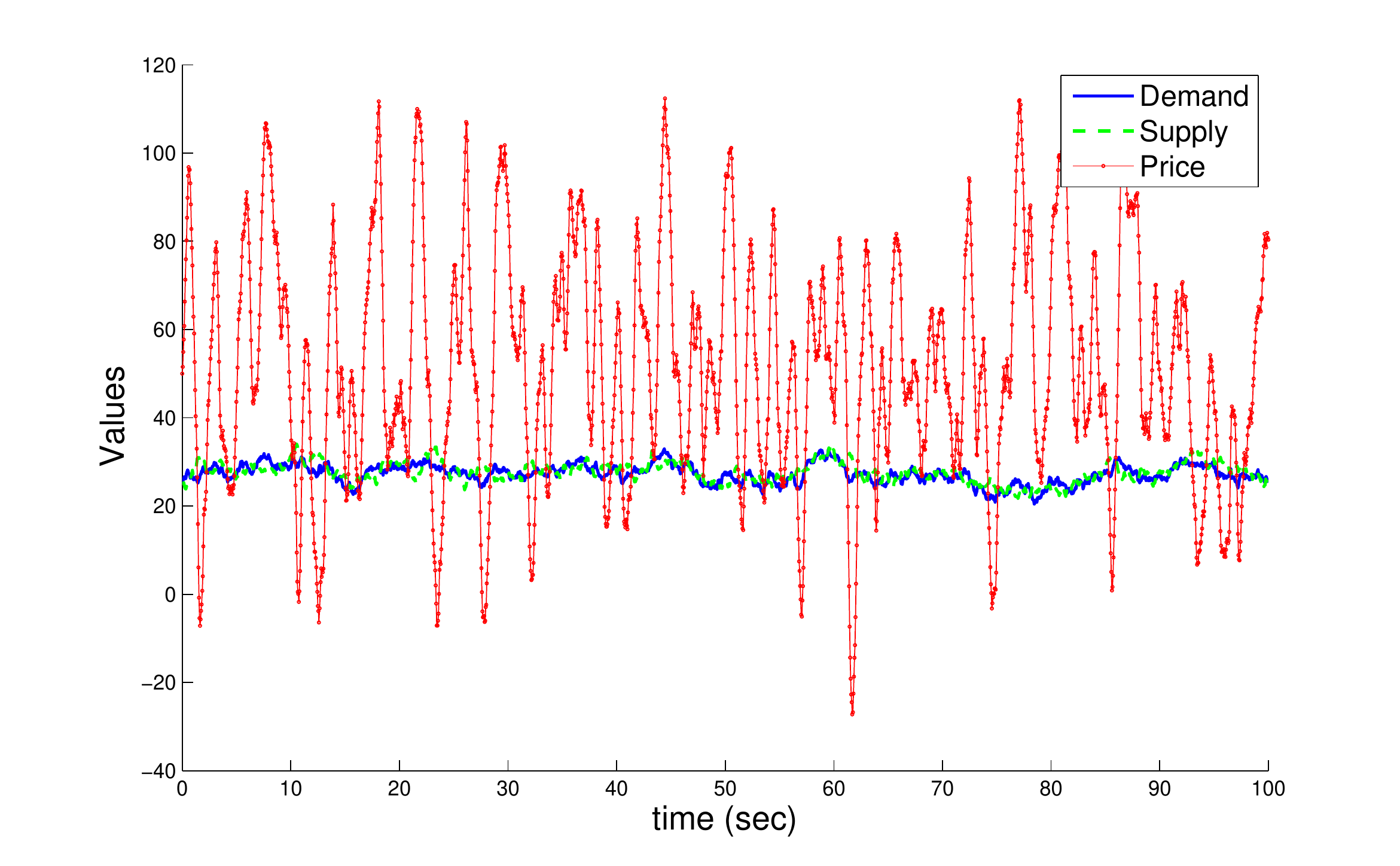}\\ \vspace{-.55cm}
\includegraphics[width=9.25cm,height=2.00cm]{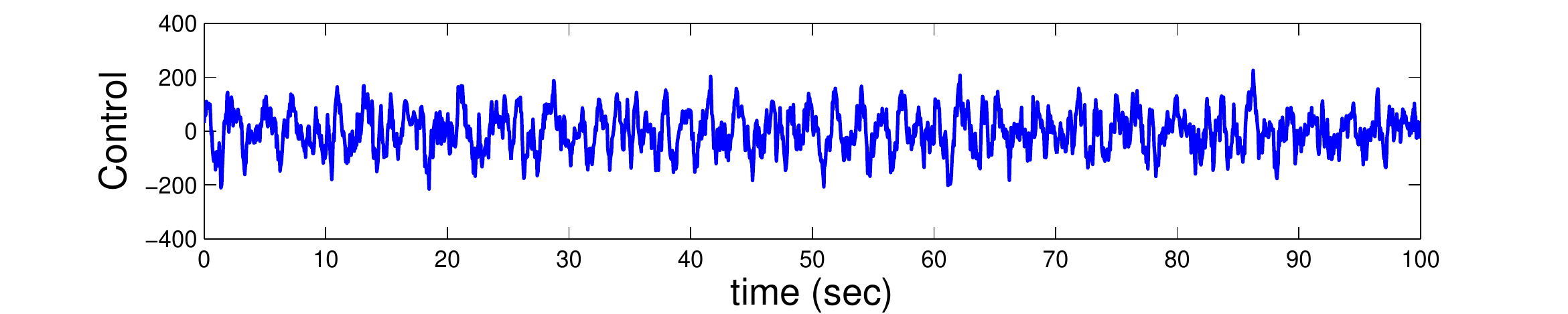}
\caption{Dynamics when r = 0.01}
\label{fig:dynamics_R01_C}
\end{minipage}
\hspace{0.5cm}
\begin{minipage}[b]{0.5\linewidth}
\centering
\includegraphics[width=9.25cm]{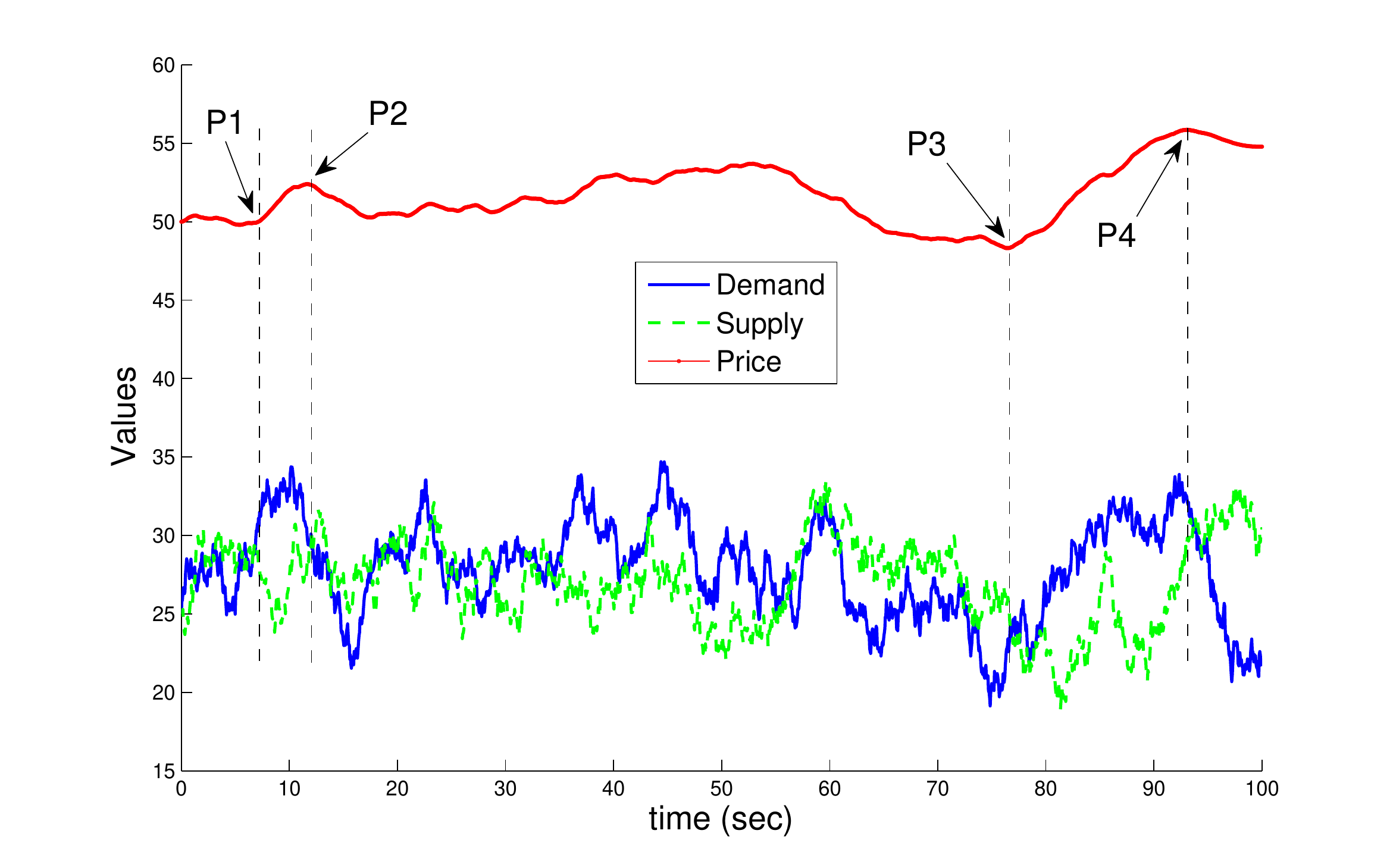}\\ \vspace{-.55cm}
\includegraphics[width=9.25cm,height=2.00cm]{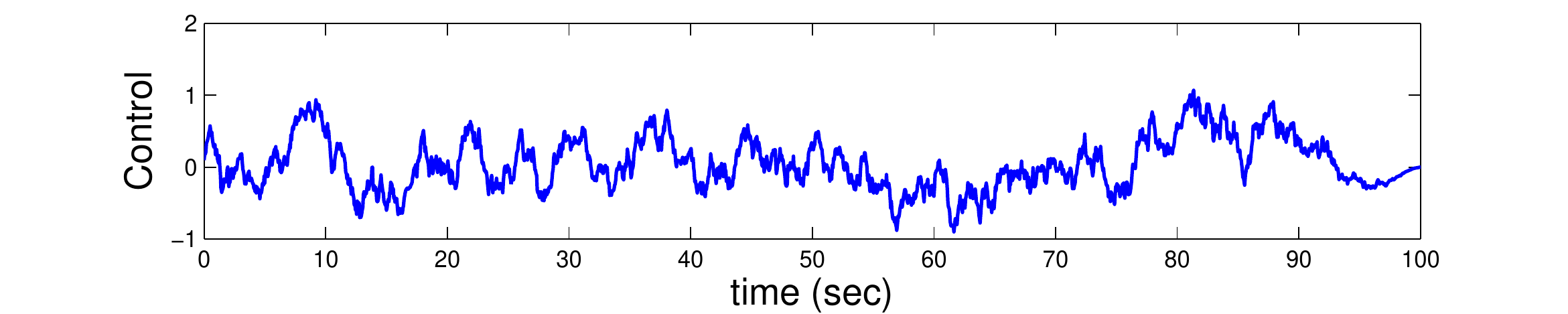}
\caption{Dynamics when r = 1000} 
\label{fig:dynamics_R1000_C}
\end{minipage} 
\end{figure}

Now we present two more simulations with $r=1$. The effect of the initial state on the trajectory is observed here. In Fig.\ \ref{fig:dynamics_dhs} initially, demand is higher than the supply, whereas in Fig.\ \ref{fig:dynamics_dls} demand is lower than the supply. As expected, the price process becomes very volatile in early stages to stabilise the market.

\begin{figure}[ht]
\begin{minipage}[b]{0.5\linewidth}
\centering
\includegraphics[width=9.25cm]{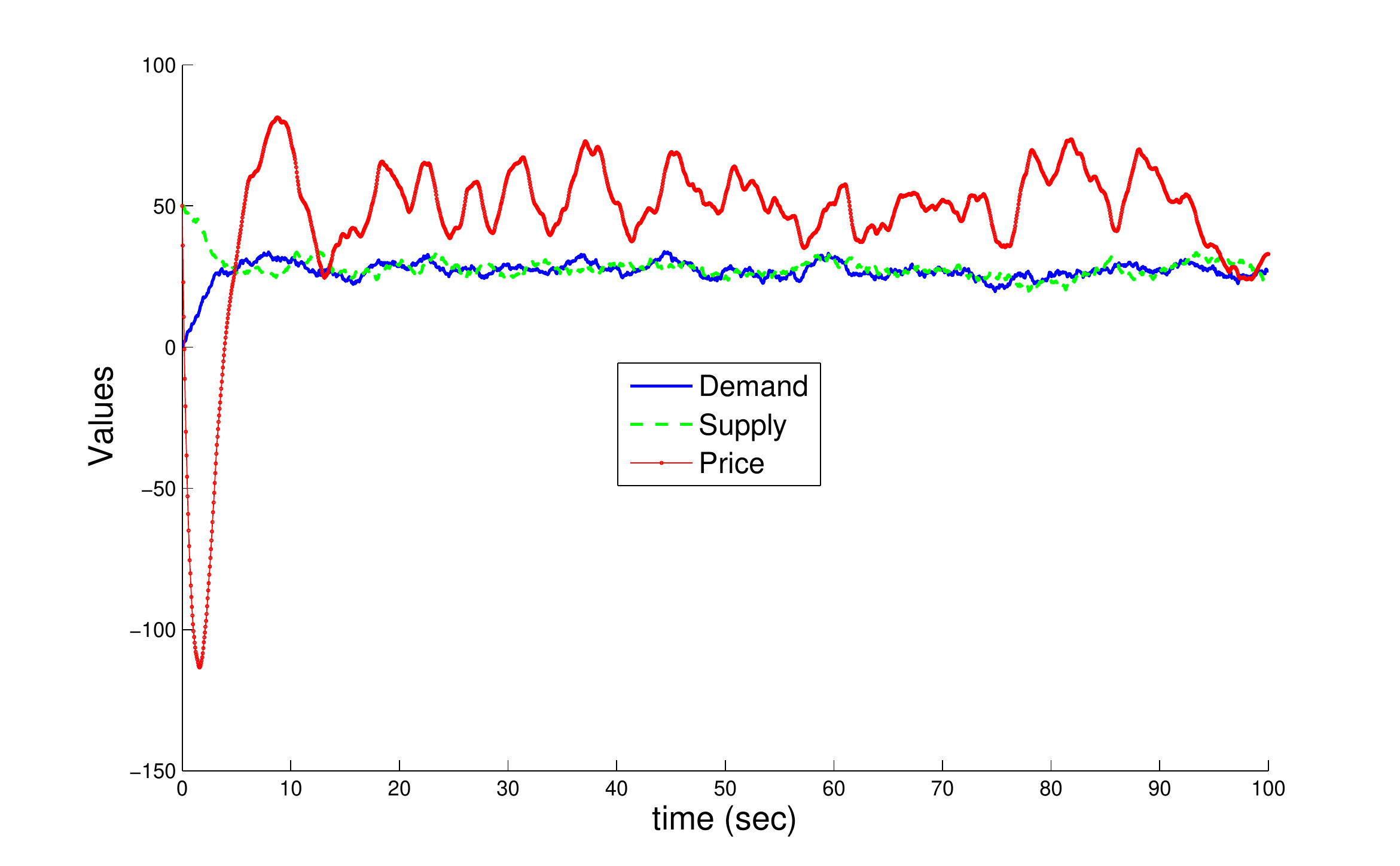}\\ \vspace{-.55cm}
\includegraphics[width=9.25cm,height=2.00cm]{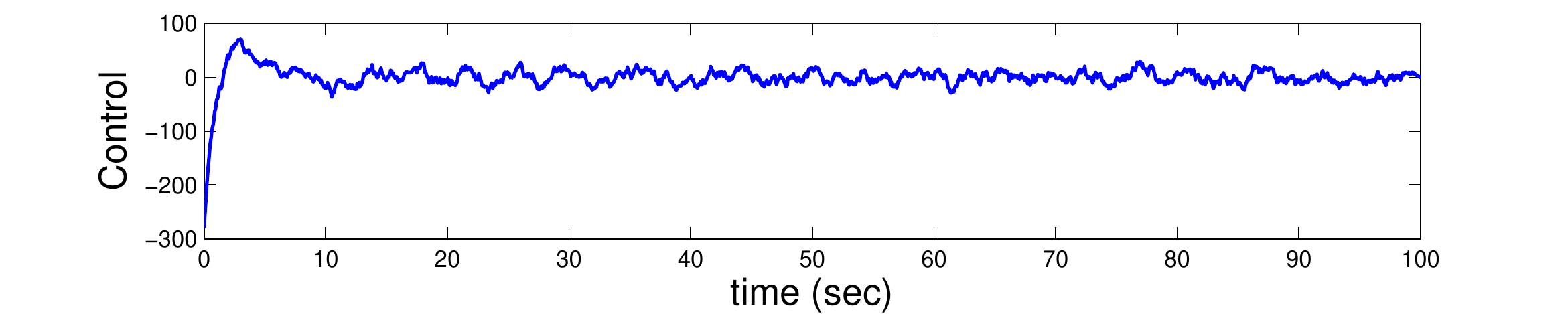}
\caption{Dynamics when initial supply is higher than demand}
\label{fig:dynamics_dhs}
\end{minipage}
\hspace{0.5cm}
\begin{minipage}[b]{0.5\linewidth}
\centering
\includegraphics[width=9.25cm]{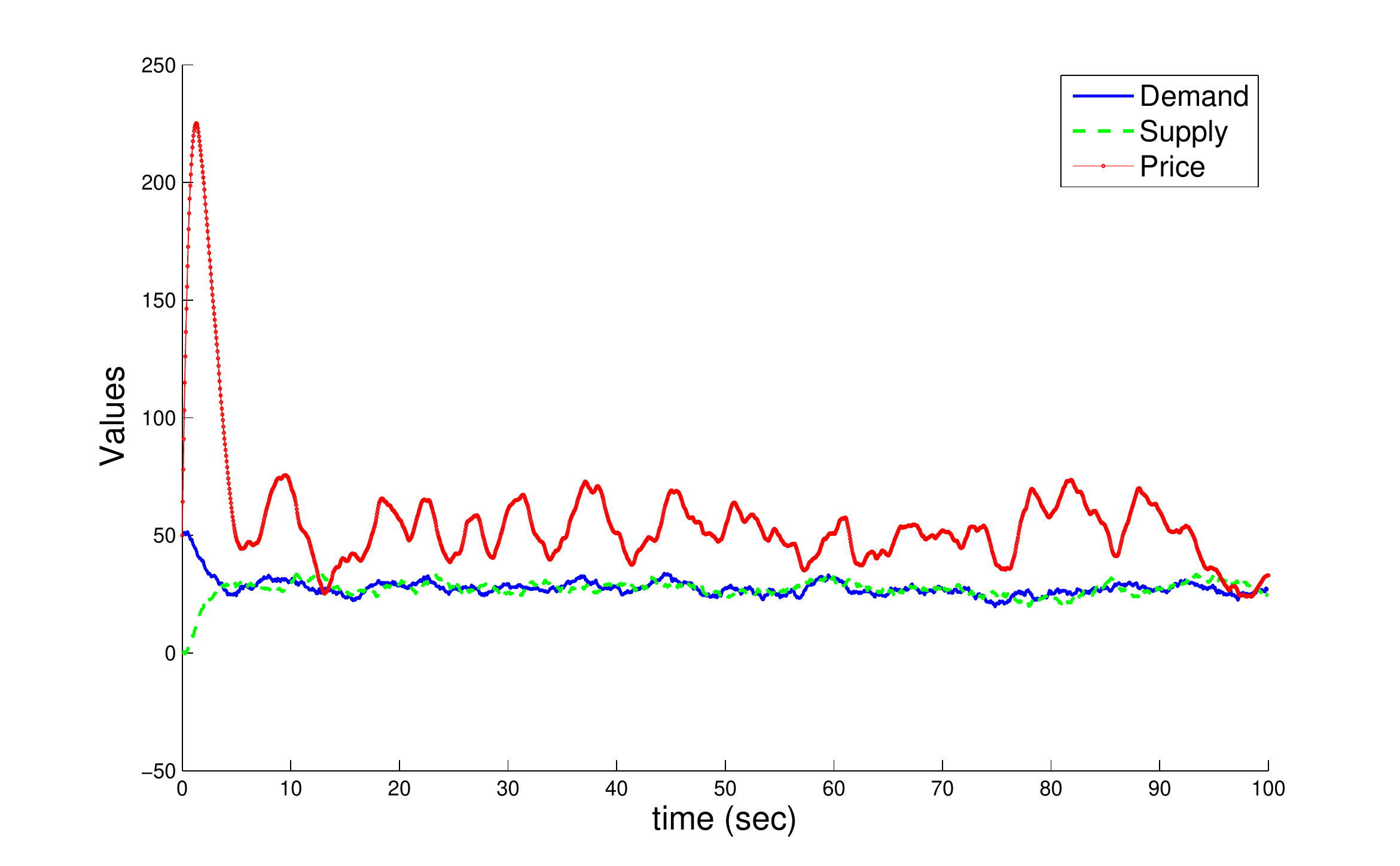}\\ \vspace{-.55cm}
\includegraphics[width=9.25cm,height=2.00cm]{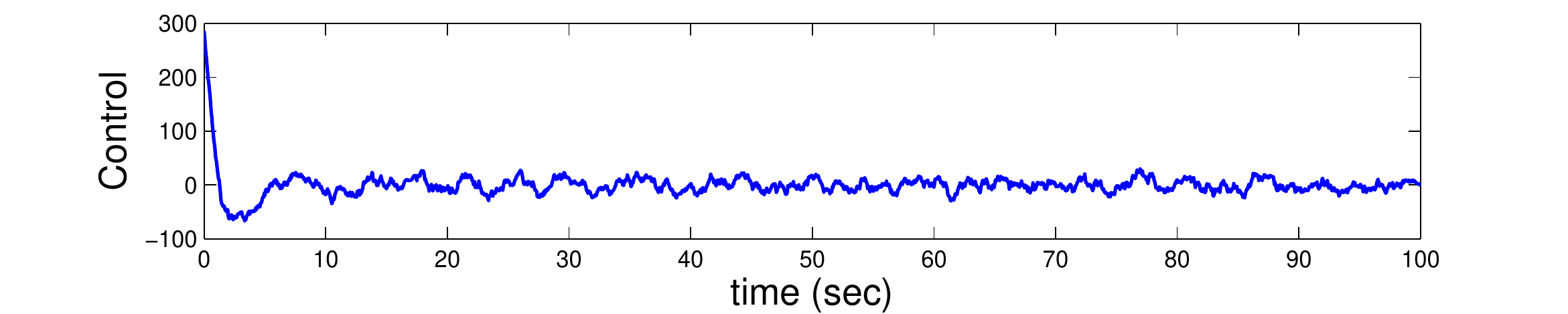}
\caption{Dynamics when initial supply is lower than demand} 
\label{fig:dynamics_dls}
\end{minipage} 
\end{figure}

Finally, we present an experimental result showing the relation between $r$ and the average absolute difference between supply and demand dynamics. Recall that high costs are paid when this difference is high, and as seen in Fig.\ \ref{fig:avg_dis}, as $r$ increases the average absolute difference increases. The $x$ axis is drawn on a logarithmic scale in order to capture the graph on lower values of $r$.

\begin{figure}
\centering
\includegraphics[width=10cm,height=6.18cm]{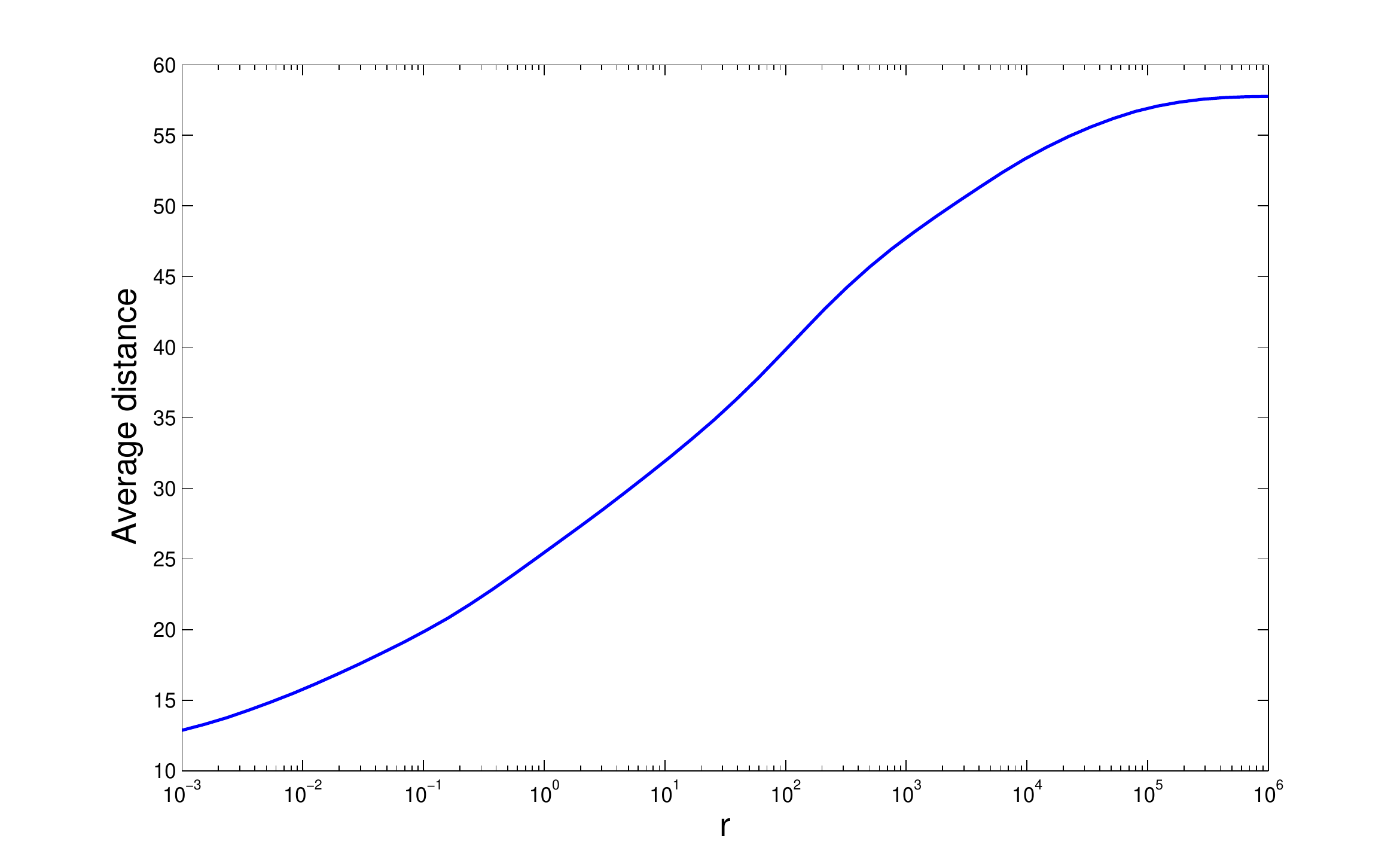}
\caption{Average absolute difference between demand and suply}
\label{fig:avg_dis}
\end{figure}

\section{DECENTRALIZED CONTROL FORMULATION}\label{sec:DCF}

We define a continuous dynamic game for $N^d$ consumers and $N^s$ suppliers. The agents continuously submit their bids as price-quantity graphs, and the system announces the resulting price. Agents buy or sell corresponding shares of supplies according to their bids. One important notion is that future demand and supply processes are dependent on the price process, which is determined instantly by the agents' price-quantity graphs shaped by their actions.

We have the set of agents $\b{N} = \{D_1,...,D_{N^d},S_1,...,S_{N^s} \}$. We define the family of three dimensional state processes $\{   (d_t^i,s_t^{d_i},p_t^{d_i})\t ; \, t\geq 0 , \, 1\leq i \leq N^d\}$ for the consumers and two dimensional state processes $\{   (s^i_t,p^{s_i}_t)\t ; \, t\geq 0 ,\, 1\leq i \leq N^s\}$ for the suppliers. The initial conditions $\{ d^i_0,s_0^{d_i},p_0^{d_i},1\leq i \leq N^d; \, s^j_0, p_0^{s_j},1 \leq j \leq N^s \}$ are mutually independently distributed bounded random variables which are independent of the standard Wiener processes $\{ w_t^{d_i},\, 1\leq i \leq N^d , \; w_t^{s_j},\, 1\leq j \leq N^s ; \; t\geq 0  \}$. The process $d_t^i$ is the demand dynamics for agent $D_i$, the process $s_t^{d_i}$ is the supply it receives, and the process $p_t^{d_i}$ is the parameter it applies to its pre-announced price-quantity graph function $\phi^{d_i}(p_t;p_t^{d_i})$. For the supplier side $s_t^i$ is the current supply and $p_t^{s_i}$ is the parameter for the price-quantity graph $\phi^{s_i}(p_t;p_t^{s_i})$. Here $\{ \phi^{d_i}, \, 1\leq i \leq N^d \}$ and $\{ \phi^{s_i}, \, 1\leq i \leq N^s\}$ are the price-quantity graphs that the consumers and the suppliers submit to the market clearing price functional $f^m(\cdot) \in \mathbf{C}_b$ for the instant price $p_t$ determination. The dynamics for the consumers and the suppliers for $t\geq 0$ are given as
\begin{equation}\label{eqn:G_dyna}
\begin{aligned}
dd_t^i = & f^{d_i}(d_t^i,p_t,\phi^{d_i}(p_t;p_t^{d_i}) )dt + \sigma_d dw_t^{d_i},   \enspace 1 \leq i \leq N^d,\\
dp_t^{d_i} = & u_t^{d_i} dt,\enspace 1 \leq i \leq N^d,\\
ds_t^i = & f^{s_i}(s_t^i,p_t,\phi^{s_i}(p_t;p_t^{s_i}))dt + \sigma_s dw_t^{s_i},  \enspace 1\leq i \leq N^s,\\
dp_t^{s_i} = & u_t^{s_i} dt, \enspace 1\leq i \leq N^s,\\
p_t = & f^m (\{\phi^{d_i}(\cdot,\cdot),\,1\leq i \leq N^d; \; \phi^{s_j}(\cdot,\cdot),\, 1\leq j \leq N^s \}).
\end{aligned}
\end{equation}

The actions of the agents $\{ u^{d_i},\, 1\leq i \leq N^d ; \; u^{s_j},\, 1\leq j \leq N^s\}$ control the size of the increments for $\{ p^{d_i},\, 1\leq i \leq N^d ; \; p^{s_j},\, 1\leq j \leq N^s\}$. The functional $ f^{d_i}, \, 1\leq i \leq N^d,$ is allowed to be a function of $d^i,\, p$ and $\phi^{d_i}(\cdot,\cdot)$, values of the demand of the consumer agent $D_i$, the price and its price-quantity graph; and $ f^{s_i}, \, 1\leq i \leq N^s,$ is allowed to be a functional of $s^i,p,\phi^{s_i}(\cdot,\cdot)$, values of the supply of the supplier $S_i$, the price and its price-quantity graph.   

Following \cite{2010CM}, the individual loss function of a consumer and a supplier are defined respectively:
\begin{equation}\label{eqn:G_inst_cost}
\begin{aligned}
g^d(\cdot) = &  p_t \cdot s_t^{d_i} - v \cdot \min(d_t^i,s_t^{d_i}) + c_{bo}(s_t^{d_i}-d_t^i), \\
g^s(\cdot)  = &  c(s_t^i)-p_t \cdot s_t^i.
\end{aligned}
\end{equation}

Finally, the cost functions associated with each consumer, each supplier and  corresponding control actions $u^{d_i},\, 1\leq i \leq N^d$, and $ u^{s_j},\, 1\leq j \leq N^s$, are specified to be

\begin{equation}\label{eqn:G_inte_cost}
\begin{aligned}
J_{d}(d_0^i,s_0^{d_i},p_0,u^{d_i}) = & \enspace  \mathbb{E} \int_0^T [ p_t \cdot s_t^{d_i} - v \cdot \min(d_t^i,s_t^{d_i}) + c_{bo}(s_t^{d_i}-d_t^i)] dt,\quad 1\leq i \leq N^d,\\
J_{s}(s_0^i,p_0,u^{s_i}) = & \enspace  \mathbb{E} \int_0^T [c(s_t^i)-p_t \cdot s_t^i] dt,\quad 1\leq i \leq N^s.
\end{aligned}
\end{equation}

We employ \ass{ass:init} for initial values and the disturbance processes, and \ass{ass:fric} on the functions $f^{d_i}(\cdot)$ and $f^{s_i}(\cdot)$. Moreover,

\begin{hypot}\label{ass:G_incr_dyna_func}
$f^{d_i}(\cdot),\, 1\leq i \leq N^d$, is a strictly decreasing function of $p$, whereas $f^{s_i}(\cdot),\, 1\leq i \leq N^s$, is strictly increasing. The price-quantity graphs for the consumers are decreasing functions of $p_t$ in the form of $\phi^{d_i}(p_t;p_t^{d_i})\triangleq f^{\phi^{d_i}}(p_t^{d_i}) - p_t$, whereas the price-quantity graphs are increasing in the form of $\phi^{s_i}(p_t;p_t^{s_i})\triangleq f^{\phi^{s_i}}(p_t^{s_i}) + p_t$, for the suppliers. Functions $f^{\phi^{d_i}}(p_t^{d_i})$ and $f^{\phi^{s_i}}(p_t^{s_i})$ are Lipschitz continuous on $\R$ with Lipschitz constants $Lip(f^{\phi^{d_i}}),\, 1\leq i \leq N^d$, and $Lip(f^{\phi^{s_i}}),\, 1\leq i \leq N^s$. 
\end{hypot}
Consequently, for some $\gamma>0,\, \eta>0$, the market clearing price function $f^m(\cdot) \in \b{C}_b:\R\rightarrow \R$ is a linear function in the form of $f^m \triangleq  (\gamma/(N^d+N^s)) \cdot  (  \sum_{i=1}^{N^d} f^{\phi^{d_i}}(\cdot) + \sum_{i=1}^{N^s} f^{\phi^{s_i}}(\cdot)  + \eta )$.

This assumption limits the model to a price process parameterized by $\gamma>0$ and $\eta>0$ obtained by price-quantity graph functions submitted by the consumer and supplier agents: $\phi^{d_i}(\cdot),\, 1\leq i \leq N^d$, $\phi^{s_i}(\cdot),\, 1\leq i \leq N^s$, that are linear functions of $p_t,t\geq 0$.

\ass{ass:OU} is employed: the demand processes $\{(d_t^i,\,t\geq 0);\; 1\leq i \leq N^d\}$ and the supply processes $\{(s_t^i,\, t\geq 0) ;\;  1\leq i \leq N^s\}$ are linear mean-reverting processes that have bounded variances.
As a special case, we consider linear quadratic functions below. The choice of quadratic terms can be explained by the convexity of the production cost and the blackout cost functions. The rest of the cost functions can be arranged in a way that fits the linear parameters of the quadratic cost functions defined below. We use a penalty function for the control actions and define the \emph{volatility coefficient} $r$. Increasing the volatility coefficient penalizes each agent's attempt to change its price-quantity functional; therefore increasing $r$ is equivalent to penalizing volatility in the market when the system of agents is taken as a mass. The nonlinear curve-fitting problem is solved in the least-squares sense given the input data and the observed output data, and we get the following cost functions:
\begin{equation}\label{eqn:G_LQG_cost}
\begin{aligned}
J_{d}(d_0^i,s_0^{d_i},p_0,u^{d_i}) = &  \mathbb{E} \int_0^T \Big [ {x_t^{d_i}}^\top Q^d x_t^{d_i} + 2{x_t^{d_i}}^\top D_t^{d_i} + r (u_t^{d_i})^2 \Big ] dt,\quad 1 \leq i \leq N^d,\\
J_{s}(s_0^i,p_0,u^{s_i}) = &  \mathbb{E} \int_0^T \Big [ {x_t^{s_i}}^\top Q^s x_t^{s_i} + 2{x_t^{s_i}}^\top D_t^{s_i} + r (u_t^{s_i})^2 \Big ] dt,\quad 1\leq i \leq N^s,
\end{aligned}
\end{equation}
where $x^{d_i} := (d^i,s^{d_i},p^{d_i})\t,\, x^{s_i} := (s^i,p^{s_i})\t$, $Q^{d,s}\geq 0$, $r>0$ are constant values, and $D_t^{d_i}$ is a continuous vector valued function of $\{x_t^{d_j},1\leq j \leq N^d, j\neq i;\, x_t^{s_j},1\leq j \leq N^s\}$ and $D_t^{s_i}$ is a continuous vector valued function of $\{x_t^{d_j},1\leq j \leq N^d;\, x_t^{s_j},1\leq j \leq N^s,j\neq i\}$. The cost functions are coupled: the price functional (dependent on all agents' actions) enters into the cost function parameters. Employing \ass{ass:OU}, the equation system \eqref{eqn:G_dyna} can be written in the form of
\begin{equation}\label{eqn:G_line_dyna}
\begin{aligned}
dx^{d_i}_t = &\psi(x_t^{d_i},u_t^{d_i})dt + G^d dw_t^{d_i},\quad t\geq 0,\\
dx^{d_i}_t = &\left ( A^{d_i} x^{d_i}_t + B^d u^{d_i}_t + h_t^{d_i} \right ) dt + G^d dw_t^{d_i},\quad t\geq 0,\\
dx^{s_i}_t = & \psi(x_t^{s_i},u_t^{s_i})dt + G^s dw_t^{s_i},\quad t\geq 0,\\
dx^{s_i}_t = &\left ( A^{s_i} x^{s_i}_t + B^s u^{s_i}_t + h_t^{s_i} \right ) dt + G^s dw_t^{s_i},\quad t\geq 0,
\end{aligned}
\end{equation}
where $\{ w_t^{d_i},\, 1\leq i \leq N^d , \; w_t^{s_i},\, 1\leq i \leq N^s ; \; t\geq 0  \}$ are standard Wiener processes with suitable dimensions and $x^{d_i}(0)=x^{d}_0,\, x^{s_i}(0)=x^{s}_0$. The function $h_t^{d_i}$ is of the form $h_t^{d_i}(p_t^{d_j},1\leq j \leq N^d, j\neq i;\, p_t^{s_j},1\leq j \leq N^s)$ and function $h_t^{s_i}$ is of the form $h_t^{s^i}(p_t^{s_j},1\leq j \leq N^s, j\neq i;\, p_t^{d_j},1\leq j \leq N^d)$.

The coefficients $[A^{d_i,s_i},B^{d,s}]\in \Th \in \R^{n(n+m)}$, will be called the \emph{dynamics parameters}. The variability of dynamics parameters from agent to agent is used to model a heterogeneous population of agents. Note that the dynamics are coupled among agents only through the price functional. The price functional enters into $h^{d_i}$ and $h^{s_i}$; and is a function on all agents' price-quantity graph functions. The following assumption is employed:
\begin{hypot}\label{ass:comp}
The set of dynamics parameters, $\Th$, is a compact set in the form of $\Th \subset \R^{n(n+m)}$.
\end{hypot}
\section{EQUILIBRIUM ANALYSIS}\label{sec:equ_ana}

In this section we first define the value functions for the consumers and the suppliers with the dynamics \eqref{eqn:G_line_dyna} and the cost functions \eqref{eqn:G_LQG_cost}. We then show the existence of suboptimal solutions to the HJB equations using the perturbation method. Secondly, we present the closed form solutions for the HJB equations and the statistical dependence among the agents through the price process. This dependence leads to an implementation issue which is overcome by a policy iteration style method applied by each agent to calculate the best response action. Finally, we show the existence of a unique subgame perfect equilibrium of the dynamic game under a fixed point argument in a system of agents where each agent applies the policy iteration method.

\subsection{Existence and Uniqueness of the Best Response Actions}

From now on, we consider \eqref{eqn:G_LQG_cost} and \eqref{eqn:G_line_dyna}. We define the admissible control set, $\mathcal{U}_3$, of each consumer and supplier as the set  of all feedback controls adapted to $\F_t$, the $\sigma$-field generated by the agents' trajectories and the price process $\{x_\tau^{d_{i}}, x_\tau^{s_{j}},p_\tau;0\leq \tau \leq t,\, 1\leq i \leq N^d,\, 1\leq j \leq N^s \}$. The minimum cost-to-go from any agent's initial state is described by the \emph{value functions} which are defined by $V^{d_i}(0,x_0^{d_i}) = \inf_{u\in \mathcal{U}_3} J_d(x_0^{d_i},u^{d_i}),\, 1\leq i \leq N^d; \; V^{s_i}(0,x_0^{s_i}) = \inf_{u\in \mathcal{U}_3} J_s(x_0^{s_i},u^{s_i}),\, 1\leq i \leq N^s$. Whenever the treatment is the same for both consumers and suppliers' value functions $V^{d_i},V^{s_i}$, we will drop the superscripts. The value function solves the Hamilton-Jacobi-Bellman (HJB) Equation:
\begin{equation}\label{eqn:G_HJB}
-\frac{\p V}{\p t}  + \sup_{u \in \mathcal{U}_3} \left \{ - \frac{\p V}{\p x}\t \psi - r(u)^2 \right \} - \frac{1}{2}\tr \left ( \frac{\p^2 V}{\p x^2}GG\t \right ) - x\t Q x - 2x\t D_t = 0,
\end{equation}
where $V(T,x) = 0$.

As discussed before in Section \ref{sec:CCF}, due to the lack of uniform parabolicity \cite{1999YZ}, classical solutions may be hard to obtain. Viscosity solutions are typically adopted in these circumstances. In order to approximate the solution we add the term $(1/2)\epsilon^2(\p^2V/\p p^2)$ to \eqref{eqn:G_HJB} and obtain uniform parabolicity. We obtain the perturbed value functions $V_p^{d_i}$ and $V_p^{s_i}$:
\begin{align}
& \nonumber -\frac{\p V_p^{d_i}}{\p t}  + \sup_{u \in \mathcal{U}_3} \left \{ - {\frac{\p V_p^{d_i}}{\p x^{d_i}}}^\top \psi^{d_i} -  r (u^{d_i})^2 \right \} - \frac{1}{2}\sigma_d^2\frac{\p^2 V_p^{d_i}}{\p d_i^2 }\\
& \label{eqn:G_per_HJB1} \hspace{3.5cm} - \frac{1}{2}\sigma_{s^d}^2\frac{\p^2 V_p^{d_i}}{\p {s^{d_i}}^2 } - \frac{1}{2}\epsilon^2 \frac{\p^2 V_p^{d_i}}{\p {p^{d_i}}^2 } - {x^{d_i}}^\top Q^d{x^{d_i}} - 2{x^{d_i}}^\top D_t^{d_i} = 0,
\end{align}
\begin{equation}\label{eqn:G_per_HJB2}
\frac{\p V_p^{s_i}}{\p t}  + \sup_{u \in \mathcal{U}_3} \left \{ - {\frac{\p V_p^{s_i}}{\p x^{s_i}}}^\top \psi^{s_i} -  r (u^{s_i})^2 \right \} - \frac{1}{2}\sigma_s^2\frac{\p^2 V_p^{s_i}}{\p s_i^2 } - \frac{1}{2}\epsilon^2 \frac{\p^2 V_p^{s_i}}{\p {p^{s_i}}^2 } - {x^{s_i}}^\top Q^s x^{s_i} - 2{x^{s_i}}^\top D_t^{s_i} = 0,
\end{equation}
where $V_p^{d_i}(T,x)=0$, and $V_p^{s_i}(T,x)=0$.

Eqs.\ \eqref{eqn:G_per_HJB1} and \eqref{eqn:G_per_HJB2} have unique solutions as stated in the following theorem.
\begin{thm}\label{thm:G_HJBperexis}
For all $\e>0$, Equations \eqref{eqn:G_per_HJB1} and \eqref{eqn:G_per_HJB2} have unique solutions for the admissible control set $\mathcal{U}_3$.
\end{thm}
\textit{Proof:} The proof is very similar to the proof of Theorem \ref{thm:CCF_HJBperexis}, therefore omitted.

In the theorem below, we prove that the solutions to the perturbed value functions $V_p^d$ and $V_p^s$ \eqref{eqn:G_per_HJB1}, \eqref{eqn:G_per_HJB2} converge uniformly to the value function $V$ obtained from the HJB Equation \eqref{eqn:G_HJB}.
\begin{thm}\label{thm:G_HJBperconv}
For $x\in \R^3$ or $x\in \R^2$ suitably and $\epsilon>0$, if we define $V_p$ as the solution to \eqref{eqn:G_per_HJB1} or \eqref{eqn:G_per_HJB2} and $V$ as the solution to \eqref{eqn:G_HJB} for the admissible control set $\mathcal{U}_3$, then $V_p\rightarrow V$ uniformly on $[0,T]$ as $\e\conv 0$.
\end{thm}
\textit{Proof:} The proof is very similar to the proof presented for Theorem \ref{thm:CCF_HJBperconv}, therefore omitted.

\subsection{Closed Form Solution}

Standard arguments \cite[Section 2.3]{1989AM} show that $J(x,u)$ is quadratic in $x$. Furthermore, at any point $x\in \mathbb{R}^3$ or suitably $x\in \mathbb{R}^2$ and $t\in [0,T]$, the minimum cost-to-go is quadratic in $x$. Consequently, $V$ is of the form $V(0,x) = x\t K(0) x + 2 x\t S(0) + q(0)$, that satisfies the boundary condition $V^{d_i}(T,x^{d_i}) = 0,\enspace \forall x^{d_i}\in \mathbb{R}^3$, and $V^{s_i}(T,x^{s_i}) = 0,\enspace \forall x^{s_i}\in \mathbb{R}^2$.

Following the same steps in Sec.\ \ref{sec:CCF_cl} we obtain
\begin{equation}\label{eqn:G_opt_control}
u^*(t) = -r^{-1}B\t [K(t)x(t)+S(t)].
\end{equation}

$K$ and $S$ in \eqref{eqn:G_opt_control} are iterated backwards in time, and depend on other agents' actions on $0\leq t \leq T$. This implies that at time $t\geq 0$, an agent can not calculate its best response simply through its own trajectory and the control action history of all the agents on $0 \leq s \leq t$. The agents are coupled through the price process, and the full trajectory of the price process needs to be calculated in order to obtain the best response action. The analysis of the best response of each agent and the corresponding equilibrium is presented in the next section.

\subsection{Subgame Perfect Equilibrium}\label{sec:spe}
In this section we analyze the equilibrium properties. At each time iteration and at each point in the state space each agent solves the ODEs \eqref{eqn:K}, \eqref{eqn:S} and \eqref{eqn:q} for all the consumers and the suppliers in the system, calculates their best response actions \eqref{eqn:G_opt_control}, and simultaneously solves these equations for each agent to obtain the unique fixed point in the action space. As stated before, the admissible control set is $\mathcal{U}_3$, the set  of all feedback controls adapted to $\F_t$, the $\sigma$-field generated by the agents' trajectories and the price process $\{x_\tau^{d_{i}}, x_\tau^{s_{j}},p_\tau;0\leq \tau \leq t,\, 1\leq i \leq N^d,\, 1\leq j \leq N^s \}$. Each individual agent knows the dynamics and cost function parameters of all agents in the system. Therefore, at a certain time $t\geq 0$, and for a given point in the state space, each agent can solve the ODEs \eqref{eqn:K}, \eqref{eqn:S} and \eqref{eqn:q} that depend on all agents, and get the unique fixed point for the action profile. The state processes of all agents in the system are statistically dependent due to the price functional; however, at each $t\geq 0$, given that all dynamics and state information is known, the best response calculations for all agents can be independently calculated by each agent in the system. We show that the system of equations regarding the best response actions of all agents in the system has a unique solution. We also present the policy iteration procedure that leads to the unique solution of the system of equations when applied by all agents in the system. Due to the stochasticity of the system dynamics \eqref{eqn:G_line_dyna}, this procedure is repeated by each agent until the fixed point is obtained at each time iteration. The compactness of the parameter set and the boundedness of the price functional ensure the existence of the fixed point. 

\begin{hypot}\label{ass:contr}
$[A^{d_i,s_i},B]\in \Th$ is controllable, $[Q^{1/2},A^{d_i,s_i}]$ is observable, and $A_*$ is a Hurwitz matrix. For all $\th\in \Th$, all the eigenvalues of $A_*(\theta) \triangleq A(\theta) - B(\theta) r^{-1}B\t(\theta) K(\theta)$ have negative real part; $A_*$ is continuous over $\Th$; there exists $\kappa>0,\, \rho>0$ such that $\lVert e^{A_*(\theta)t}\rVert \leq \kappa e^{-\rho t},\, \forall t \geq 0$.
\end{hypot}
The closed form solution is written for $S$ and $K$ as
\begin{align}
\label{eqn:loop_S} S(t) = & \int_t^{T} e^{-A_*\t(t-\tau)}K(\tau)h(p(\tau))d\tau + \int_t^{T} e^{-A_*\t(t-\tau)}D(p(\tau))d\tau \triangleq \mathcal{T}_1 p_t,\\
\nonumber K(t) \triangleq & \enspace \mathcal{T}_2 p_t.
\end{align}
where $\mathcal{T}_2$ is the solution to the Riccati equation \eqref{eqn:K}.

Since the solution $S(t;\th),\, \th \in \Th$, to the ordinary differential equation \eqref{eqn:S}, and the solution $K(t;\th),\, \th\in \Th$, to the Riccati equation \eqref{eqn:K} parameterized by $\th \in \Th$ are smooth functions of $\th$ (see \cite{1984De_TAC}), $S(t;\th)$ and $K(t;\th)$ satisfy the following lemma.

\begin{lem}\label{lem:S_K_contr}
Under \ass{ass:init}, \ass{ass:OU}, \ass{ass:G_incr_dyna_func}, \ass{ass:comp} and \ass{ass:contr}, we have $\mathcal{T}_1 p \in \b{C}_b[0,\infty)$ and $\mathcal{T}_2 p \in \b{C}_b[0,\infty)$ for any $p(\cdot) \in \b{C}_b[0,\infty)$.   
\end{lem}

Therefore $\mathcal{T}_1$ and $\mathcal{T}_2$ are bounded continuous maps. With the best response actions applied, assuming $x_0^{d_i} = 0,1\leq i \leq N^d;\; x_0^{s_i}=0,1\leq i \leq N^s$, without loss of generality, in the closed loop we have
\begin{equation}\label{eqn:loop_x}
\begin{aligned}
\mathbb{E} x_t^{d_i}(t) = &- \int_0^t e^{A_*^d (t-\tau)} B^d r^{-1} {B^d}\t S^{d_i}(p(\tau)) d\tau + \int_0^t e^{A_*^d(t-\tau)} h^{d_i}(p(\tau)) d\tau \triangleq \mathcal{T}_3 p_t,\\
\mathbb{E} x_t^{s_i}(t) = & - \int_0^t e^{A_*^s(t-\tau)} B^s r^{-1} {B^s}\t S^{s_i}(p(\tau)) d\tau + \int_0^t e^{A_*^s(t-\tau)} h^{s_i}(p(\tau)) d\tau \triangleq  \mathcal{T}_4 p_t.    
\end{aligned}
\end{equation}
\begin{lem}
Under \ass{ass:init}, \ass{ass:OU}, \ass{ass:G_incr_dyna_func}, \ass{ass:comp} and \ass{ass:contr}, $\mathcal{T}_3 p \in \b{C}_b[0,\infty)$ and $\mathcal{T}_4 p \in \b{C}_b[0,\infty)$  for any $p(\cdot)\in \b{C}_b[0,\infty)$.   
\end{lem}
\begin{IEEEproof}
Due to \ass{ass:contr}, $A_*$ is a Hurwitz matrix. Moreover we have shown in Lemma \ref{lem:S_K_contr} that $S(t;\theta)$ is a bounded value in $ \b{C}_b[0,\infty)$. Therefore  $\mathcal{T}_3 p \in \b{C}_b[0,\infty)$  and  $\mathcal{T}_4 p \in \b{C}_b[0,\infty)$ follows.
\end{IEEEproof}

Now, we write the price function $f^m(\cdot)$ for $t\geq 0$:
\begin{equation}\label{eqn:G_contr} 
p_t =  f^m (\{\phi^{d_i}(p_t;p_t^{d_i}),\,1\leq i \leq N^d; \; \phi^{s_i}(p_t;p_t^{s_i}),1\leq i \leq N^s \}) \enspace \triangleq \enspace \mathcal{T}_5 p_t. 
\end{equation}

The following lemma establishes that $\mathcal{T}_5$ defined above is a map from $\b{C}_b[0,\infty)$ to itself. 

\begin{lem}
Under \ass{ass:init}, \ass{ass:OU}, \ass{ass:G_incr_dyna_func}, \ass{ass:comp} and \ass{ass:contr}, we have  $\mathcal{T}_5 p \in \b{C}_b[0,\infty)$  for any $p\in \b{C}_b[0,\infty)$.   
\end{lem}

In the following theorem we show that $\mathcal{T}_5$ has a fixed point. Following that we deduce that in a system of consumers and suppliers with dynamics \eqref{eqn:G_line_dyna} and cost functions \eqref{eqn:G_LQG_cost} the system has a unique equilibrium, and this equilibrium is indeed the unique subgame perfect equilibrium within Markovian strategies. At this point we introduce the following technical assumption:

\begin{hypot}\label{ass:ctrcn} \vspace{-.8cm}
\begin{equation*}
\begin{aligned}
& \frac{\gamma\kappa^2}{N^d+N^s} \Bigg ( N^d M_{Lip(f^{\phi^{d}})} \cdot \bigg ( \frac{1}{\rho^2} \lVert B^d \rVert^2  r^{-1} \Big ( M_{K^d} M_{h^{d}} + M_{D^{d}} \Big ) + \frac{1}{\rho} M_{h^{d}} \bigg)\\ 
& \qquadthree + N^s M_{Lip(f^{\phi^{s}})} \cdot \bigg ( \frac{1}{\rho^2} \lVert B^s \rVert^2  r^{-1} \Big ( M_{K^s}  M_{h^{s}} +  M_{D^{s}} \Big ) + \frac{1}{\rho} M_{h^{s}}  \bigg) \Bigg ) < 1,  
\end{aligned}
\end{equation*}
where $\gamma>0$ is specified in \ass{ass:G_incr_dyna_func}, $\kappa>0$ is specified in \ass{ass:contr}, $\lVert K^d(t) \rVert \leq M_{K^d}$, $\lVert K^s(t) \rVert \leq M_{K^s}$ for all $0\leq t \leq T$, $M_{h^d} = \max_{1\leq i \leq N^d}  \lVert h^{d_i} \rVert$, $M_{h^s} = \max_{1\leq i \leq N^s}  \lVert h^{s_i} \rVert$, $M_{D^d} = \max_{1\leq i \leq N^d}  \lVert D^{d_i} \rVert$, $M_{D^s} = \max_{1\leq i \leq N^s}  \lVert D^{s_i} \rVert$, $M_{Lip(f^{\phi^{d}})} = \max_{1\leq i \leq N^d}  \lVert Lip(f^{\phi^{d_i}}) \rVert$, $M_{Lip(f^{\phi^{s}})} = \max_{1\leq i \leq N^s}  \lVert Lip(f^{\phi^{s_i}}) \rVert$, $Lip(f^{\phi^d})$ and $Lip(f^{\phi^s})$ are the Lipschitz constants respectively for $f^{\phi^{d_i}}(\cdot)$ and $f^{\phi^{s_i}}(\cdot)$ specified in \ass{ass:G_incr_dyna_func}.    
\end{hypot}

This technical assumption ensures the uniqueness of the price process. Note that there is a trade-off between $\gamma$ and $r$ in this inequality. A small $r$ means cheaper control actions for the agents; therefore, the price process is more likely to be volatile and eventually intractable. This assumption ensures tractability even for small values of $r$. Note that numerical results show that this assumption can be satisfied.  

\begin{thm}\label{thm:pfixed}
Under \ass{ass:init}, \ass{ass:OU}, \ass{ass:G_incr_dyna_func}, \ass{ass:comp}, \ass{ass:contr} and \ass{ass:ctrcn}, the map $\mathcal{T}_5:\b{C}_b[0,\infty)\rightarrow \b{C}_b[0,\infty)$ has a unique fixed point which is uniformly Lipschitz continuous on $[0,\infty)$.
\end{thm}
\textit{Proof:} The proof is given in Appendix \ref{thm:pfixed_pro}.

The main result of this section immediately follows Theorem \ref{thm:pfixed}.
\begin{cor}\label{cor:punique}
Under \ass{ass:init}, \ass{ass:OU}, \ass{ass:G_incr_dyna_func}, \ass{ass:comp}, \ass{ass:contr} and \ass{ass:ctrcn}, the expected value of the equation system \eqref{eqn:G_dyna} admits a unique bounded solution. 
\end{cor}

\subsection{Policy Iteration} 
We now describe the iterative policy of an agent from its policy space. At time $t\in [0,T]$, for a fixed iteration number $k\geq 0$ and $\tau \in [t,T]$, suppose that there is a priori $p_\tau(k)\in \b{C}_b[0,\infty)$. Then the best response action \eqref{eqn:G_opt_control} of each agent is in the form of $u_\tau^*(k+1) = -r^{-1}B\t [K_\tau(k+1)x_\tau(k)+S_\tau(k+1)]$. Taking the same steps in the previous section we get the recursion for $p_\tau(k)$ as $\mathbb{E}[p_\tau(k+1)] = \mathcal{T}_5 p_\tau(k)$. The procedure can be applied for all $t \leq \tau \leq T$, and the recursion converges to a unique $p^*(\tau),\, t \leq \tau \leq T$, and once the price trajectory is obtained, each agent is able to calculate its best response action \eqref{eqn:G_opt_control}. The existence and uniqueness of $p^*(\tau),\, t \leq \tau \leq T$, are shown in Theorem \ref{thm:pfixed} by use of a fixed point argument. This procedure is independently performed by each agent in the system at each time iteration. The following Proposition may be proved by Theorem \ref{thm:pfixed}. 
\begin{prop}\label{prop:pol_itr}
Under \ass{ass:init}, \ass{ass:OU}, \ass{ass:G_incr_dyna_func}, \ass{ass:comp}, \ass{ass:contr} and \ass{ass:ctrcn}, $\lim_{k\rightarrow \infty} p_\tau (k) =p^*(\tau)$ for any $p^*\in [0,\infty)$ where $p^*$ is the solution to \eqref{eqn:G_dyna}.
\end{prop}

Before we present the subgame perfect equilibrium theorem, we employ the assumption below:

\begin{hypot}\label{ass:punish}
Agents can only use Markovian strategies, i.e., we rule out many non-myopic subgame perfect equilibria mostly based on future punishments such as \emph{grim-trigger strategies} in repeated and dynamic games \cite{1979Ru_ET}. 
\end{hypot}

A Markovian strategy $\gamma_i$ of a player is defined to be a strategy where for each $t$, $\gamma_i(t,x)$ depends on $\F_t$, the $\sigma$-field generated by the agents' trajectories and the price process $\{x_\tau^{d_{i}}, x_\tau^{s_{j}},p_\tau;0\leq \tau \leq t,\, 1\leq i \leq N^d,\, 1\leq j \leq N^s \}$ only through $t$, $\{x_t^{d_{i}}, x_t^{s_{j}},p_t;\, 1\leq i \leq N^d,\, 1\leq j \leq N^s \}$.

Let us define $\Gamma$ to be the class of mappings $\gamma:[0,T]\times \R^3 \rightarrow \R$ with the property that $u(t)=\gamma(t,x)$ is adapted to $\F_t$, the $\sigma$-field generated by the agents' trajectories and the price process $\{x_\tau^{d_{i}}, x_\tau^{s_{j}},p_\tau;0\leq \tau \leq t,\, 1\leq i \leq N^d,\, 1\leq j \leq N^s \}$. A \emph{subgame perfect equilibrium} of the dynamic game with the set of agents $\b{N}$ with dynamics \eqref{eqn:G_dyna}, and with the cost functions \eqref{eqn:G_inte_cost} is a strategy profile $\gamma^* \in \Gamma$ such that for any history $h$, the strategy profile $\gamma^*|_h$ is a Nash equilibrium of the subgame based on the history $h$. 

Under \ass{ass:punish} the iterative update of agents' policies results in the system's unique subgame perfect equilibrium.

\begin{cor}
Under \ass{ass:fric}, \ass{ass:init}-\ass{ass:punish}, for agents with dynamics \eqref{eqn:G_dyna}, the action profile obtained when all agents apply \eqref{eqn:G_opt_control} at any $t\geq 0$, with the iterative procedure described above is the unique subgame perfect equilibrium of the game.
\end{cor}

The system has a unique equilibrium within Markovian strategies, and the action profile corresponding to the equilibrium can be obtained by an iterative algorithm applied by each agent. The equilibrium is shown by use of a fixed point argument, and the procedure is explained by a policy iteration methodology. At each time iteration each agent looks at the future, estimates the price trajectory, and calculates the best response action. This procedure leads to the unique best response profile, and the equilibrium is obtained.

\subsection{Efficiency--Volatility Trade-off}\label{sec:eff_vol}

We would like to look at the relation between the social cost function and the action penalizing parameter, \emph{volatility coefficient}, $r$. The social cost function is defined as $J \triangleq J (\{d^i,s^{d_i},s^j,p,u^{d_i},u^{s_j}$, $1\leq i \leq N^d,\, 1\leq j \leq N^s\})  =   \sum_{i=1}^{N^d} J_{d}(d^i,s^{d_i},p,u^{d_i}) + \sum_{i=1}^{N^s} J_{s}(s^i,p,u^{s_i})$,
\begin{equation}\label{eqn:G_s_cost}
\begin{aligned}
J =  \sum_{i=1}^{N^d} \mathbb{E} \int_0^T \left ( {x_t^{d_i}}^\top Q^d x_t^{d_i} + 2{x_t^{d_i}}^\top D_t^{d_i} + r (u_t^{d_i})^2 \right ) dt + \sum_{i=1}^{N^s} \mathbb{E} \int_0^T \left ( {x_t^{s_i}}^\top Q^s x_t^{s_i} + 2{x_t^{s_i}}^\top D_t^{s_i} + r (u_t^{s_i})^2 \right ) dt.
\end{aligned}
\end{equation}

We define \emph{efficiency} as the quantity obtained when the social cost is multiplied by $-1$. \emph{Volatility} on the other hand is defined by the price fluctuation measured by $\sum_{i=1}^{N^d} \mathbb{E} \int_0^T (u_t^{d_i})^2 dt + \sum_{i=1}^{N^s} \mathbb{E} \int_0^T (u_t^{s_i})^2 dt$.

Then we remove $r u^2$ from the cost function and define the \emph{state penalizing social cost} under the best response actions \eqref{eqn:G_opt_control} applied both by the consumers and the suppliers. Note that the social cost is merely a summation of the same type of cost functions; therefore the analysis of a single summand term can be easily carried onto the whole cost function. The state penalizing social cost function is written:%
\begin{equation}\label{eqn:G_spcost}
J_{sp}^* \enspace \triangleq \enspace \sum_{i=1}^{N^d} \mathbb{E} \int_0^T \left ( {x_t^{d_i}}^\top Q^d x_t^{d_i} + 2{x_t^{d_i}}^\top D_t^{d_i}  \right ) dt + \sum_{i=1}^{N^s} \mathbb{E} \int_0^T \left ( {x_t^{s_i}}^\top Q^s x_t^{s_i} + 2{x_t^{s_i}}^\top D_t^{s_i}  \right ) dt.
\end{equation}

\begin{thm}\label{thm:G_spc}
For all $x\in\R^3$, or $x\in\R^2$ suitably, the state penalizing cost portion \eqref{eqn:G_spcost} of the cost function \eqref{eqn:G_s_cost} using the best response solution $u^*$ is an increasing function of $r$.
\end{thm}
\textit{Proof:} The proof is very similar to the proof of Theorem \ref{thm:spc}; therefore not provided.

\begin{cor}\label{thm:G}
Suppose \ass{ass:fric}, \ass{ass:init}-\ass{ass:punish} hold. Let the price adjustment of the consumers and the suppliers be penalized with a factor of the \emph{volatility coefficient}. Increasing the \emph{volatility coefficient} increases the integrated social cost, while decreasing the coefficient decreases the cost. 

Therefore, there is an inherent trade-off between social efficiency and non-volatility.
\end{cor}

\subsection{Simulations}\label{sec:sim}

Here we simulate a power market. We use Euler-Maruyama Method \cite{1974Ar} for discretization of the stochastic differential equations. The dynamics equations for $1\leq i \leq N^d,\, 1\leq j \leq N^s$ are $d_{k+1}^i = d_k^i - \rho\left ( d_k^i-(\beta-p_k) \right ) \Delta t + \sigma w_k^{d_i}\sqrt{\Delta t},\enspace p_{k+1}^{d_i} = p_k^{d_i} + u_k^{d_i} \Delta t,\enspace s_{k+1}^j = s_k^j - \rho\left ( s_k^j-(p_k-p_k^{s_j}) \right ) \Delta t + \sigma w_k^{s_j}\sqrt{\Delta t},\enspace p_{k+1}^{s_j} = p_k^{s_j} + u_k^{s_j} \Delta t,\enspace p_k = (\sum_{1\leq l \leq N^d}p_k^{d_l}+\sum_{1\leq l \leq N^s}p_k^{s_l})/(N^d+N^s)$, where $\rho=0.05, \Delta t = 0.05, \beta = 75,\sigma=2,t_{final}=500$, with the initial conditions $(d_0^{d_i},p_0^{d_i})\t=(25,50)\t,\;(s_0^{s_i},p_0^{s_i})\t=(25,50)\t, \; p_0 = 50$.

We present a couple of figures showing the dynamics when $r=0.005$ and $r=100$. The high volatility of the price in Fig.\ \ref{fig:dynamics_R005} compared to the low volatility of the price in Fig.\ \ref{fig:G_dynamics_R1000} can be observed. One can also notice the effect of volatility on stability. In Fig.\ \ref{fig:dynamics_R005}, the aggregate demand and aggregate supply dynamics follow a much closer path, whereas in Fig.\ \ref{fig:G_dynamics_R1000}, the gap between these two processes turns out to be fluctuating. As the highest costs are paid when the absolute difference between the aggregate demand and supply is the highest, the social cost paid in Fig.\ \ref{fig:G_dynamics_R1000} is larger than in Fig.\ \ref{fig:dynamics_R005}.

\begin{figure}[ht]
\begin{minipage}[b]{0.5\linewidth}
\centering
\includegraphics[width=8.35cm]{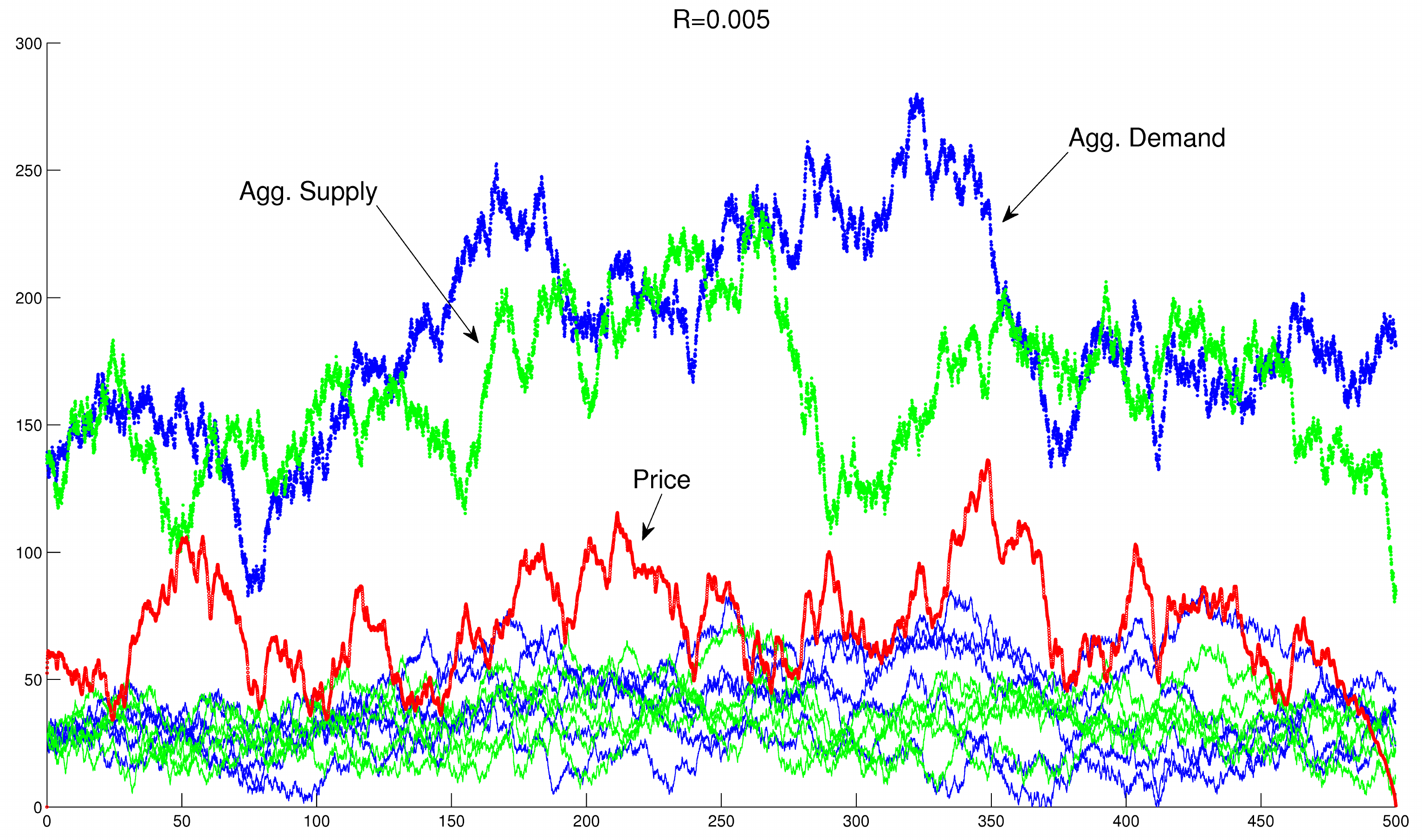}
\caption{r=0.005}
\label{fig:dynamics_R005}
\end{minipage}
\begin{minipage}[b]{0.5\linewidth}
\centering
\includegraphics[width=8.35cm]{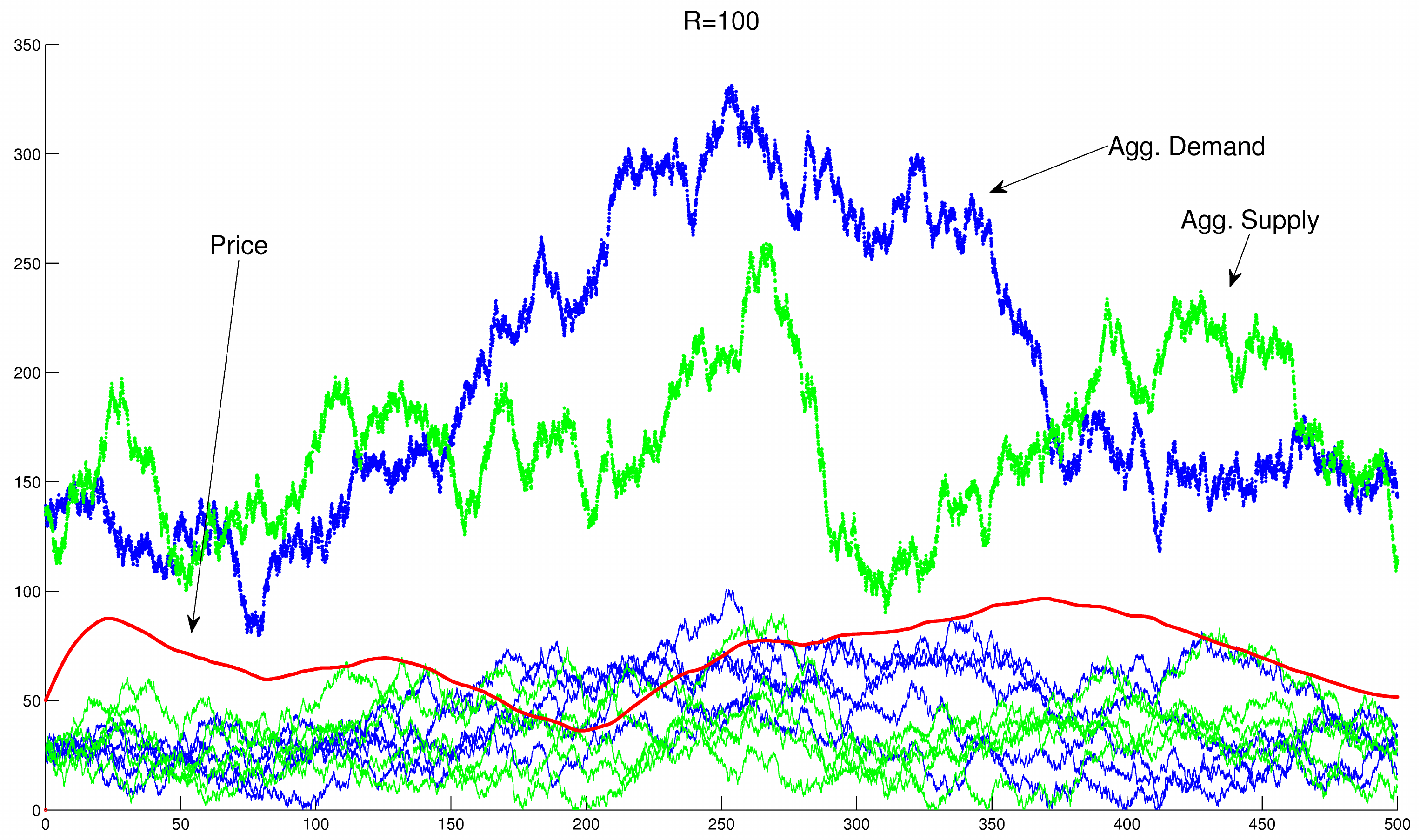}
\caption{r=100} 
\label{fig:G_dynamics_R1000}
\end{minipage} 
\end{figure}

\section{CONCLUSION}\label{sec:Conc}

The \textit{no friction} assumption of the traditional market mechanisms is a significant factor in terms of analyzing dynamic continuous time markets. Friction adds complexity to the analysis of the market, and volatility is an inevitable result of friction as our stylized model shows.

The market model has been defined as a social cost optimization problem for the regulator, and it has been shown that a simple linear closed form control exists. As a special case we studied quadratic cost functions and have shown that there is a trade-off between social efficiency of the market and non-volatility. High efficiency requires a volatile market, whereas one has to compromise efficiency in order to get stability in prices. Then we have extended the optimization model to a dynamic game theoretic framework and have shown the same efficiency--volatility trade-off for this model with closed form best response actions. The complexity increases as the number of agents in the system increases; therefore, this leads to an implementation problem. In future work we will apply mean field methods and study the limit behaviour of a large population model of suppliers and consumers to precompute the price process, therefore decrease the computational complexity.

Also, in the current model, it is assumed that each agent in the system knows all the dynamical and cost function parameters of other agents in the system. Adaptation and learning methods will be applied in future work for a system of agents that start with partial statistical information on other agents' dynamical and cost function parameters. The adapting (or learning) agents will be allowed to only observe a random partial fraction of the population of agents' trajectories and the price process. Convergence rates will be studied and properties of the equilibrium will be investigated in this setup. The behaviour of the \emph{efficiency--volatility} theorem will also be analyzed.

\appendices
\section{PROOF OF THEOREM \ref{thm:CCF_exis}} \label{thm:CCF_exis_pro}

\begin{IEEEproof}

The existence of the solution follows a standard argument \cite[Chapter 2, Theorem 5.1]{1999YZ}. The continuity of $V(x)$ is based on the continuous dependence of the cost function on the initial state values.

The plan of the proof is as follows: first, to show the uniqueness of the control, we show that the cost function $J$ is strictly convex in $p$. By contradiction: we show that the probability of two distinct control action trajectories leading to the same price process $p$ is of Lebesgue measure zero. This proves uniqueness.

We have 
\[ J(t,x_t,u) = \E \int_t^T [ -v \min(d_\tau,s_\tau)+c(s_\tau)+c_{bo}(d_\tau-s_\tau) ]d\tau,\]
which can be written for $\Delta t > 0$
\[ J(t,x_t,u) = [-v \min(d_t,s_t)+c(s_t)+c_{bo}(d_t-s_t)]\Delta t + \E\int_{\Delta t}^T (.)d\tau.\]
The first part is independent of $p_t$. We have
\[ J(t,x_t,u) = (\cdot) + \E\int_{\Delta t}^T [-v \min(d_\tau,s_\tau)+c(s_\tau)+c_{bo}(d_\tau-s_\tau)]d\tau.\]
For $s=\tau-\Delta t$, 
\begin{align*}
& J(t,x_t,u) = (\cdot) + \E \int_0^{T-\Delta t} [-v \min(d_s+f^d(d_s,p_s)\Delta t + dw_s^d,s_s+f^s(s_s,p_s)\Delta t + dw_s^s) +\\
& \qquad c(s_s+f^s(s_s,p_s)\Delta t + dw_s^s)+c_{bo}(d_s+f^d(d_s,p_s)\Delta t + dw_s^d-s_s-f^s(s_s,p_s)\Delta t - dw_s^s)]ds.
\end{align*}
Therefore, $J(t,x_t,u) = (\cdot) + \E \int_0^{T-\Delta t} [ g_{s+\Delta t} ]ds$.
For $d_{s+\Delta t} < s_{s+\Delta t}$, as $p$ increases, $\E g_{s+\Delta t}$ also increases and tends to $\infty$ as $p\conv\f$. As $p$ decreases, $\E g_{s+\Delta t}$ decreases and tends to $g^*$ as $p\conv p^*$, where $\E (d_{s+\Delta t} = s_{s+\Delta t})$. For $d_{s+\Delta t} > s_{s+\Delta t}$, as $p$ decreases, $g_{s+\Delta t}$ increases and tends to $\infty$ as $p\conv\f$. As $p$ increases, $g_{s+\Delta t}$ decreases and tends to $g^*$ as $p\conv p^*$, where $\E (d_{s+\Delta t} = s_{s+\Delta t})$. 

We have $J(d_s,s_s,\frac{1}{2}(\hat{p}_s+\tilde{p}_s),u)\leq \frac{1}{2}[J(d_s,s_s,\hat{p}_s,u) + J(d_s,s_s,\tilde{p}_s,u)]$, and the inequality holds on $A\triangleq\{(s,\omega),\tilde{p}_s\neq\hat{p}_s\}$. Let $\mathbb{E}\int_0^T 1_{\tilde{p}_s\neq\hat{p}_s} ds>0$, i.e., $A$ has a strictly positive measure. Then the control $(1/2)(\hat{u}+\tilde{u})\in\mathcal{U}$ gives
\begin{align*}
J\left (x_0,\frac{1}{2}(\hat{u}+\tilde{u}) \right ) < & \frac{1}{2} [J(x_0,\hat{u}) + J(x_0,\tilde{u})],\\
= & \inf_{u\in\mathcal{U}}J(x_0,u).
\end{align*} 
which is a contradiction. Therefore, $\mathbb{E}\int_0^T 1_{\tilde{p}_s\neq\hat{p}_s} ds=0$. The trajectories of $p_s$ are continuous with probability 1 by \eqref{eqn:control_cont}. Hence, we have $\tilde{p}_s - \hat{p}_s \equiv 0$ on $[0,T]$ with probability 1. We have $\int_0^{s}(\tilde{u}_\tau-\hat{u}_\tau)d\tau =\tilde{p}_s-\hat{p}_s$, for all $s\in[0,T]$, by \eqref{eqn:control_cont}. Hence, with probability 1, $\tilde{u}_s-\hat{u}_s=0$ a.e. on $[0,T]$ which is equivalent to $\mathbb{E}\int_0^T 1_{\tilde{u}_s\neq\hat{u}_s}ds =\int_0^T\mathbb{P}_{\Omega}(\tilde{u}_s\neq\hat{u}_s)ds=0$. Consequently, $\mathbb{P}_\Omega(\tilde{u}_s\neq\hat{u}_s)>0$ only of Lebesgue measure zero on $s\in[0,T]$. Therefore, uniqueness is proved.
\end{IEEEproof}

\section{PROOF OF THEOREM \ref{thm:spc}}\label{thm:spc_pro}
\begin{IEEEproof}
If we were to only find $d J(x,u^*)/d r$, we would calculate it from
\begin{equation*}
\frac{ d J(x,u^*)}{d r} =x_0\t \frac{d K(0)}{d r} x_0 + 2 x_0\t \frac{d S(0)}{d r} + \frac{d q(0)}{d r}. 
\end{equation*}
However, we are interested in $dJ_{sp}(x,u^*)/dr$.

The quadratic cost function was shown \eqref{eqn:NFL_int_app_cost} to be $J(x,u^*) = \int_0^T [ {x^*}\t (t)Qx^*(t) + 2{x^*}\t(t)D + r (u^*(t))^2  ] dt$. We seek to compute $dJ_{sp}(x,u^*)/dr$. However, the calculations are easier for $dJ_{sp}(x,u^*)/d\gamma$, where $\gamma=r^{-1}$. We have
\begin{equation}\label{eqn:NFL_dJdt}
\frac{d J_{sp}(x,u^*)}{d t}\big |_{t=T} = f^J(x) =\\  
\left( {x^*}\t(t,\gamma)Q{x^*}(t,\gamma) + 2 {x^*}\t(t,\gamma) D\right) \big |_{t=T},
\end{equation}
with the initial condition $ J_{sp}(x,u^*)|_{t=0}=0$.
We take the derivative of \eqref{eqn:NFL_dJdt} with respect to $\gamma$ and changing the order of differentiation
\begin{multline}\label{eqn:NFL_ddtdJdR}
\frac{d}{dt}\left( \frac{d J_{sp}(x,u^*)}{d\gamma} \right ) = \left ( \frac{d x^*(t,\gamma)}{d\gamma} \right )\t Q x^*(t,\gamma) \\
+ {x^*}\t (t,\gamma) Q \left ( \frac{dx^*(t,\gamma)}{d\gamma} \right ) + 2 \left (\frac{dx^*(t,\gamma)}{d\gamma} \right)\t D
, \quad \frac{d J_{sp}(x_0,u^*)}{d\gamma}=0.
\end{multline}

In order to get \eqref{eqn:NFL_ddtdJdR}, we need  $d x^*(t,\gamma)/ d \gamma$. Using \eqref{eqn:lin_dynamics} and \eqref{eqn:NFL_opt_control} we obtain
\begin{multline}\label{eqn:NFL_dxdt}
dx^*(t,\gamma) = f^x(x^*(t,\gamma),\gamma,K(t,\gamma),S(t,\gamma))dt + Gd\omega =\\ 
\left ( Ax^*(t,\gamma) - B \gamma B\t \big ( K(t,\gamma)x^*(t,\gamma) + S(t,\gamma) \big ) + h \right ) dt + G d\omega ,\quad x(0)=0.
\end{multline}

Taking the derivative of \eqref{eqn:NFL_dxdt} with respect to $\gamma$ and changing the order of differentiation
\begin{equation}\label{eqn:NFL_ddtdxdR}
\begin{aligned}
& \frac{d}{dt}\left( \frac{d x^*(t,\gamma)}{d\gamma} \right ) = A\left(\frac{dx^*(t,\gamma)}{d\gamma}\right) - B\gamma B\t K(t,\gamma) \left(\frac{dx^*(t,\gamma)}{d\gamma}\right)- B\gamma B\t\frac{dK(t,\gamma)}{d\gamma}x^*(t,\gamma)\\ 
& \qquadtwo  - BB\t K(t,\gamma)x^*(t,\gamma) - B\gamma B\t \frac{dS(t,\gamma)}{d\gamma} - BB\t S(t,\gamma) ,\quad \frac{d x(0)}{d\gamma}=0.
\end{aligned}
\end{equation}

In order to calculate \eqref{eqn:NFL_ddtdxdR}, we need to calculate $ d K(t,\gamma) /d\gamma = f^K(K(t,\gamma),\gamma)$ and\\ $ d S(t,\gamma)/d \gamma=f^S(K(t,\gamma),\gamma,S(t,\gamma))$.

The differential equation for $K(t,\gamma),\enspace 0\leq t \leq T$, and $S(t,\gamma),\enspace 0\leq t \leq T$, were shown to be \eqref{eqn:K} and \eqref{eqn:S}. Using the same changing order of differention technique, we get:
\begin{align}
\label{eqn:NFL_ddtdKdR} \frac{d}{dt}\left( \frac{d K(t,\gamma)}{d\gamma} \right ) = & \frac{\p f^K(\cdot)}{\p K(t,\gamma)} \frac{d K(t,\gamma)}{ d \gamma} + \frac{\p f^K(\cdot)}{\p \gamma}, \quad \frac{d K(T)}{d\gamma}=0,\\
\label{eqn:NFL_ddtdSdR} \frac{d}{dt}\left( \frac{d S(t,\gamma)}{d\gamma} \right ) = & \frac{\p f^S(\cdot)}{\p S(t,\gamma)} \frac{d S(t,\gamma)}{ d \gamma} + \frac{\p f^S(\cdot)}{\p \gamma} + \frac{\p f^S(\cdot)}{\p K(t,\gamma)} \frac{d K(t,\gamma)}{ d \gamma},\quad \frac{d S(T)}{d\gamma}=0.
\end{align}

The differential equations \eqref{eqn:NFL_ddtdKdR} and \eqref{eqn:NFL_ddtdSdR} are expanded as
\begin{equation}\label{eqn:NFL_ndKdR}
\begin{aligned}
& \frac{d}{dt}\left ( \frac{d K(t,\gamma)}{d\gamma} \right )= -\left(\frac{dK(t,\gamma)}{d\gamma}\right )A - A\t \left(\frac{dK(t,\gamma)}{d\gamma}\right ) + \left(\frac{dK(t,\gamma)}{d\gamma}\right ) B \gamma B\t K(t,\gamma) \\
& \qquadtwo  + K(t,\gamma)B\gamma B\t \left(\frac{dK(t,\gamma)}{d\gamma}\right ) + K(t,\gamma)BB\t K(t,\gamma) ; \text{ and }
\end{aligned}
\end{equation}
\begin{equation}\label{eqn:NFL_ndSdR}
\begin{aligned}
& \frac{d}{dt} \left(\frac{d S(t,\gamma)} {d\gamma} \right) = -A\t \frac{dS(t,\gamma)}{d\gamma} + (B\gamma B\t K(t,\gamma))\t \frac{dS(t,\gamma)}{d\gamma} \\
& \qquadtwo + \left(\frac{dK(t,\gamma)}{d\gamma}\right)\t B\gamma B\t S(t,\gamma) + K\t(t) B B\t S(t,\gamma) - \frac{dK(t,\gamma)}{d\gamma}h.
\end{aligned}
\end{equation}

Therefore we have all the components to calculate $dx^*(t,\gamma)/d\gamma$ \eqref{eqn:NFL_ddtdxdR}. Consequently, we are ready to calculate $dJ_{sp}(x,u^*)/d\gamma$ \eqref{eqn:NFL_ddtdJdR}. One can write
\begin{equation}\label{eqn:NFL_JgR}
\frac{dJ_{sp}(x,u^*)}{d\gamma} = \frac{dJ_{sp}(x,u^*)}{d(f(r))}= \frac{dJ_{sp}(x,u^*)}{f'(\cdot)dr} = \frac{dJ_{sp}(x,u^*)}{dr}\frac{-1}{r^{-2}},\text{ where }\gamma = r^{-1}.
\end{equation}

From \eqref{eqn:NFL_JgR} one can calculate $dJ_{sp}(x,u^*)/dr$. Here we use the specific nature of matrices $A$ and $B$ given in \eqref{eqn:ABG} and show analytically that $dJ_{sp}(x,u^*)/dr > 0$, and it is an increasing function of time. Here we show the derivations calculated at the steady state values of $K$, $S$, and $x$. However, the calculations also hold for any $0<t\leq T$. %
At steady state only the terms involving the noise variables effect $dJ_{sp}(x,u^*)/d\gamma$. Apart from that for all the other terms $dJ_{sp}(x,u^*)/d\gamma=0$. We solve the stochastic differential equation \eqref{eqn:NFL_ddtdxdR} and obtain
\begin{equation}\label{eqn:sto_xs}
\begin{aligned}
& \frac{dx^*(t,\gamma)}{d\gamma} = \int_0^t e^{\left (A-B\gamma B\t K(\tau,\gamma)\right )(t-\tau)}\bigg [ \Big(-B\gamma B\t\frac{dK(\tau,\gamma)}{d\gamma}-BB\t K(\tau,\gamma)\Big)x^*(\tau,\gamma)\\
& \qquadeight -B\gamma B\t \frac{dS(\tau,\gamma)}{d\gamma}-BB\t S(\tau,\gamma) \bigg ] d\tau.
\end{aligned}
\end{equation}

The stochastic differential equation for $x^*$ in \eqref{eqn:sto_xs} gives the solution
\[
x^*(\tau,\gamma) = \int_0^\tau e^{\left (A-B\gamma B\t K(s,\gamma)\right )(\tau-s)}(-B\gamma B\t S(s,\gamma) + h) +  \int_0^\tau e^{\left (A-B\gamma B\t K(s,\gamma)\right )(\tau-s)} G dW(s).
\]

We now inject these two equations into \eqref{eqn:NFL_ddtdJdR}, solve the stochastic differential equation, arrange the terms, and obtain
\begin{equation}\label{eqn:R1R2}
\begin{aligned}
& R_1(t) = \int_s^t e^{\left(A-B\gamma B\t K(\tau,\gamma)\right )\t(\tau-s)}\Big (-B\gamma B\t\frac{dK(\tau,\gamma)}{d\gamma}-BB\t K(\tau,\gamma)\Big)\t \cdot\\ 
& \qquadtwo e^{\left (A-B\gamma B\t K(\tau,\gamma)\right )\t(s-\tau)}d\tau,\\
& R_2(t) = \int_0^t e^{\left(A-B\gamma B\t K(\tau,\gamma)\right )(t-s)}GG\t R_1(t) e^{\left(A-B\gamma B\t K(\tau,\gamma)\right )\t(t-s)}ds,\\
& \frac{dJ_{sp}(x,u^*)}{d\gamma} = 2\int_0^T \tr(R_2(t)*Q) dt.
\end{aligned}
\end{equation}

All the terms in \eqref{eqn:R1R2} are positive except for $(-B\gamma B\t dK(\tau,\gamma)/d\gamma-BB\t K(\tau,\gamma))\t$. The matrix $BB\t$ is an all zeros matrix except for $1$ at the rightbottom entry. The multiplication with $dK(\tau,\gamma)/d\gamma$ and $K(\tau,\gamma)$ give positive values due to \eqref{eqn:ABG} that solve \eqref{eqn:NFL_ddtdKdR} and \eqref{eqn:K}. Hence, $R_1(t)<0$.

Therefore, one obtains $dJ_{sp}(x,u^*)/d\gamma < 0$ for all $\gamma>0$. As $dJ_{sp}(x,u^*)/dr =  - (dJ_{sp}(x,u^*)/d\gamma) r^{-2}$, one obtains $dJ_{sp}(x,u^*)/dr >0$ for all $r>0$.

Also, one can see the role of the noise variance on the trade-off. $G$ in \eqref{eqn:R1R2} is the noise variance matrix; and notice that an increase in the noise variance leads to an increase in $dJ_{sp}^*/dr$.

Hence, the state penalizing cost portion \eqref{eqn:spcost} of the cost function \eqref{eqn:NFL_int_app_cost} in the closed loop using the optimal control $u^*$ is an increasing function of $r$. Thus, Theorem \ref{thm:spc} is proved.
\end{IEEEproof}

\section{PROOF OF THEOREM \ref{thm:pfixed}}\label{thm:pfixed_pro}
\begin{IEEEproof}
We rewrite the operator $\mathcal{T}_5$ \eqref{eqn:G_contr} here again:
\begin{equation*}
\begin{aligned} 
p_t \; = \; & f^m (\{\phi^{d_i}(p_t;p_t^{d_i}),\,1\leq i \leq N^d; \; \phi^{s_i}(p_t;p_t^{s_i}),1\leq i \leq N^s \})\\
p_t \; \triangleq \; & \gamma \cdot \left ( \frac{ \sum_{i=1}^{N^d} f^{\phi^{d_i}}(\cdot) + \sum_{i=1}^{N^s} f^{\phi^{s_i}}(\cdot) }{N^d+N^s} + \eta \right).
\end{aligned}
\end{equation*}

$\mathcal{T}_5$ is a map from the Banach space $\b{C}_b[0,\infty)$ to itself. For any $x,y\in \b{C}_b[0,\infty)$, injecting $\mathbb{E}x^{d_i},\mathbb{E}x^{s_i}$ from \eqref{eqn:loop_x} and $S(\cdot)$ from \eqref{eqn:loop_S} we have
\begin{align*}
\lVert ( \mathcal{T}_5 x - \mathcal{T}_5 y ) (t) \rVert & \\
& \hspace{-2cm} \leq \frac{\gamma\lVert x - y \rVert_\infty}{N^d+N^s} \cdot \Bigg (  \sum_{i=1}^{N^d} f^{\phi^{d_i}} \bigg ( - \int_0^t e^{A_*^d (t-\tau)} B^d r^{-1} {B^d}\t \cdot\\  
& \hspace{-1cm} \bigg ( \int_t^{T} e^{A_*^d (s-t)}K^d (s) h^{d_i} ds + \int_t^{T} e^{A_*^d(s-t)}D^{d_i} ds \bigg ) d\tau   + \int_0^t e^{A_*^d(t-\tau)} h^{d_i}  d\tau \bigg)\\  
& \hspace{-1cm} + \sum_{i=1}^{N^s} f^{\phi^{s_i}} \bigg ( \int_0^t e^{A_*^s (t-\tau)} B^s r^{-1} {B^s}\t \bigg ( \int_t^{T} e^{A_*^s (s-t)} K^s (s) h^{s_i} ds  \\
& \hspace{-1cm} + \int_t^{T} e^{A_*^s (s-t)}D^{s_i} ds \bigg ) d\tau- \int_0^t e^{A_*^s (t-\tau)} h^{s_i}  d\tau \bigg)  \Bigg).
\end{align*}
Then we solve the integrals and obtain
\begin{align*}
\lVert ( \mathcal{T}_5 x - \mathcal{T}_5 y ) (t) \rVert & \\
& \hspace{-2cm} \leq \frac{\gamma\kappa^2\lVert x - y \rVert_\infty}{N^d+N^s} \Bigg ( \sum_{i=1}^{N^d} f^{\phi^{d_i}} \bigg ( \frac{1}{\rho^2} \lVert B^d \rVert^2  r^{-1} \Big ( M_{K^d} \lVert h^{d_i} \rVert + \lVert D^{d_i} \rVert \Big ) + \frac{1}{\rho} \lVert h^{d_i} \rVert \bigg)\\ 
& \hspace{-1cm} + \sum_{i=1}^{N^s} f^{\phi^{s_i}} \bigg ( \frac{1}{\rho^2} \lVert B^s \rVert^2  r^{-1} \Big ( M_{K^s} \lVert h^{s_i} \rVert + \lVert D^{s_i} \rVert \Big ) + \frac{1}{\rho} \lVert h^{s_i} \rVert \bigg) \Bigg ),  
\end{align*}
where $\lVert e^{A_*(\theta)t}\rVert \leq \kappa e^{-\rho t},\, \forall t \geq 0$ due to \ass{ass:contr}, $\lVert K^d(t) \rVert \leq M_{K^d}$, $\lVert K^s(t) \rVert \leq M_{K^s}$ for all $0\leq t \leq T$, $M_{h^d} = \max_{1\leq i \leq N^d}  \lVert h^{d_i} \rVert$, $M_{h^s} = \max_{1\leq i \leq N^s}  \lVert h^{s_i} \rVert$, $M_{D^d} = \max_{1\leq i \leq N^d}  \lVert D^{d_i} \rVert$, $M_{D^s} = \max_{1\leq i \leq N^s}  \lVert D^{s_i} \rVert$. All the bounds given above exist due to \ass{ass:comp}.

Now employing \ass{ass:G_incr_dyna_func} we have the bound below using the Lipschitz continuity of $f^{\phi^{d_i}},\, 1\leq i \leq N^d$ and $f^{\phi^{s_i}},\, 1\leq i \leq N^s$ with Lipschitz constants $Lip(f^{\phi^{d_i}}),\, 1\leq i \leq N^d$ and\\ $Lip(f^{\phi^{s_i}}),\, 1\leq i \leq N^s$:
\begin{align*}    
\lVert ( \mathcal{T}_5 x - \mathcal{T}_5 y ) (t) \rVert & \\
& \hspace{-2cm} \leq \frac{\gamma\kappa^2 \lVert x - y \rVert_\infty}{N^d+N^s} \Bigg ( \sum_{i=1}^{N^d} Lip(f^{\phi^{d_i}}) \cdot \bigg ( \frac{1}{\rho^2} \lVert B^d \rVert^2  r^{-1} \Big ( M_{K^d} \lVert h^{d_i} \rVert + \lVert D^{d_i} \rVert \Big ) + \frac{1}{\rho} \lVert h^{d_i} \rVert \bigg)\\ 
& \hspace{-1cm} + \sum_{i=1}^{N^s} Lip(f^{\phi^{s_i}})\cdot \bigg ( \frac{1}{\rho^2} \lVert B^s \rVert^2  r^{-1} \Big ( M_{K^s} \lVert h^{s_i} \rVert + \lVert D^{s_i} \rVert \Big ) + \frac{1}{\rho} \lVert h^{s_i} \rVert \bigg) \Bigg ).
\end{align*}
Now we insert $M_{Lip(f^{\phi^{d}})} = \max_{1\leq i \leq N^d}  \lVert Lip(f^{\phi^{d_i}}) \rVert$, $M_{Lip(f^{\phi^{s}})} = \max_{1\leq i \leq N^s}  \lVert Lip(f^{\phi^{s_i}}) \rVert$ and obtain
\begin{align*}
\lVert ( \mathcal{T}_5 x - \mathcal{T}_5 y ) (t) \rVert & \\    
& \hspace{-2cm} \leq \frac{\gamma\kappa^2\lVert x - y \rVert_\infty}{N^d+N^s} \Bigg ( N^d M_{Lip(f^{\phi^{d}})} \cdot \bigg ( \frac{1}{\rho^2} \lVert B^d \rVert^2  r^{-1} \Big ( M_{K^d} M_{h^{d}} + M_{D^{d}} \Big ) + \frac{1}{\rho} M_{h^{d}} \bigg)\\ 
& \hspace{-1cm} + N^s M_{Lip(f^{\phi^{s}})} \cdot \bigg ( \frac{1}{\rho^2} \lVert B^s \rVert^2  r^{-1} \Big ( M_{K^s}  M_{h^{s}} +  M_{D^{s}} \Big ) + \frac{1}{\rho} M_{h^{s}}  \bigg) \Bigg ).  
\end{align*}

Then from \ass{ass:ctrcn} it follows that $\mathcal{T}_5$ is a contraction and therefore has a unique fixed point $p\in \b{C}_b[0,\infty)$.

\end{IEEEproof}

\bibliographystyle{./IEEEtran} 
\bibliography{./IEEEabrv,./master}

\begin{thebibliography}{10}
\providecommand{\url}[1]{#1}
\csname url@samestyle\endcsname
\providecommand{\newblock}{\relax}
\providecommand{\bibinfo}[2]{#2}
\providecommand{\BIBentrySTDinterwordspacing}{\spaceskip=0pt\relax}
\providecommand{\BIBentryALTinterwordstretchfactor}{4}
\providecommand{\BIBentryALTinterwordspacing}{\spaceskip=\fontdimen2\font plus
\BIBentryALTinterwordstretchfactor\fontdimen3\font minus
  \fontdimen4\font\relax}
\providecommand{\BIBforeignlanguage}[2]{{%
\expandafter\ifx\csname l@#1\endcsname\relax
\typeout{** WARNING: IEEEtran.bst: No hyphenation pattern has been}%
\typeout{** loaded for the language `#1'. Using the pattern for}%
\typeout{** the default language instead.}%
\else
\language=\csname l@#1\endcsname
\fi
#2}}
\providecommand{\BIBdecl}{\relax}
\BIBdecl

\bibitem{2006Jo}
P.~L. Joskow, \emph{Electricity Market Reform: An International
  Perspective}.\hskip 1em plus 0.5em minus 0.4em\relax Elsevier, 2006, ch.
  Introduction to Electricity Sector Liberalization: Lessons Learned from
  Cross-Country Studies, pp. 1--32.

\bibitem{2005Ro}
S.~Robinson, ``Math model explains high prices in electricity markets,''
  \emph{Siam News}, pp. 4--9, 2005.

\bibitem{1998WNS}
F.~A. Wolak, R.~Nordhaus, and C.~Shapiro, ``Preliminary report on the operation
  of the ancillary services markets of the {C}alifornia independent system
  operator ({ISO}),'' 1998.

\bibitem{2003NS}
P.~Navarro and M.~Shames, ``Aftershocks - and essential lessons - from the
  {C}alifornia electricity debacle,'' \emph{Electricity Journal}, pp. 24--30,
  2003.

\bibitem{2010CM}
I.-K. Cho and S.~P. Meyn, ``Efficiency and marginal cost pricing in dynamic
  competitive markets with friction,'' \emph{Theoretical Economics}, vol.~5,
  no.~2, pp. 215--239, 2010.

\bibitem{1992Ho}
W.~Hogan, ``Contract networks for electric power transmission,'' \emph{Journal
  of Regulatory Economics}, vol.~4, pp. 211--242, 1992.

\bibitem{1973BS}
F.~Black and M.~Scholes, ``The pricing of options and corporate liabilities,''
  \emph{J. of Political Economy}, vol.~81, no.~3, pp. 637--659, 1973.

\bibitem{1993He}
S.~L. Heston, ``A closed-form solution for options with stochastic volatility
  with applications to bond and currency options,'' \emph{The Review of
  Financial Studies}, vol.~6, no.~2, pp. 327--343, 1993.

\bibitem{2002HKLW}
P.~Hagan, D.~Kumar, A.~Lesniewski, and D.~Woodward, ``Managing smile risk,''
  \emph{Wilmott Mag.}, vol.~1, pp. 84--108, 2002.

\bibitem{1986Bo}
T.~Bollerslev, ``Generalized autoregressive conditional heteroskedasticity,''
  \emph{Journal of Econometrics}, vol.~31, no.~3, pp. 307--327, 1986.

\bibitem{2008Ma}
E.~T. Mansur, ``Measuring welfare in restructured electricity markets,''
  \emph{The Review of Economics and Statistics}, vol.~90, no.~2, pp. 369--386,
  May 2008.

\bibitem{2009GS_ET}
A.~Garcia and E.~Stacchetti, ``Investment dynamics in electricity markets,''
  \emph{Journal of Economic Theory}, pp. 1--39, 2009.

\bibitem{2010MPWKS_CDC}
S.~P. Meyn, M.~Negrete-Pincetic, G.~Wang, A.~Kowli, and E.~Shafieepoorfard,
  ``The value of volatile resources in electricity markets,'' in \emph{49th
  IEEE Conference on Decision and Control}, 2010, pp. 1029--1036.

\bibitem{2010KM_All}
A.~C. Kizilkale and S.~Mannor, ``Volatility and efficiency in markets with
  friction,'' in \emph{48th Annual Allerton Conference on Communication,
  Control, and Computing}, October 2010, pp. 50 --57.

\bibitem{2004HCM_TAC}
M.~Huang, P.~E. Caines, and R.~P. Malham{\'e}, ``Uplink power adjustment in
  wireless communication systems: A stochastic control analysis,'' \emph{IEEE
  Transactions on Automatic Control}, vol.~49, pp. 1693--1708, Oct. 2004.

\bibitem{1999BO}
T.~Basar and G.~J. Olsder, \emph{Dynamic Noncooperative Game Theory}, 2nd~ed.,
  ser. SIAM Classics in Applied Mathematics.\hskip 1em plus 0.5em minus
  0.4em\relax Philadelphia, PA: Academic Press, 1999.

\bibitem{1999YZ}
J.~Yong and X.~Y. Zhou, \emph{Stochastic Controls: Hamiltonian Systems and
  {HJB} Equations}.\hskip 1em plus 0.5em minus 0.4em\relax New York:
  Springer-Verlag, 1999.

\bibitem{2005HCM_SIAM}
M.~Huang, P.~E. Caines, and R.~P. Malham{\'e}, ``Degenerate stochastic control
  problems with exponential costs and weakly coupled dynamics: Viscosity
  solutions and a maximum principle,'' \emph{SIAM Journal on Control and
  Optimization}, vol.~44, no.~1, pp. 367--387, 2005.

\bibitem{1975FR}
W.~H. Fleming and R.~W. Rishel, \emph{Deterministic and Stochastic Optimal
  Control}.\hskip 1em plus 0.5em minus 0.4em\relax Springer-Verlag, 1975.

\bibitem{1972GS}
A.~V. Skorokhod and I.~I. Gikhman, \emph{Stochastic Differential
  Equations}.\hskip 1em plus 0.5em minus 0.4em\relax Berlin: Springer-Verlag,
  1972.

\bibitem{1989AM}
B.~Anderson and J.~B. Moore, \emph{Optimal Control}.\hskip 1em plus 0.5em minus
  0.4em\relax Englewood Cliffs, NJ: Prentice-Hall, 1989.

\bibitem{1974Ar}
L.~Arnold, \emph{Stochastic Differential Equations, Theory and
  Applications}.\hskip 1em plus 0.5em minus 0.4em\relax New York: Wiley, 1974.

\bibitem{1984De_TAC}
D.~Delchamps, ``Analytic feedback control and the algebraic riccati equation,''
  \emph{IEEE Transactions on Automatic Control}, vol.~29, no.~11, pp.
  1031--1033, Nov. 1984.

\bibitem{1979Ru_ET}
A.~Rubinstein, ``Equilibrium in supergames with the overtaking criterion,''
  \emph{Journal of Economic Theory}, vol.~1, no.~1, pp. 1--9, 1979.

\end{thebibliography}

\end{document}